\documentclass[final]{amsart}

\usepackage{amsthm,amsmath,amsfonts,amssymb}
\usepackage{graphicx}
\usepackage[comma, authoryear]{natbib}
\usepackage[colorlinks,citecolor=blue,urlcolor=blue]{hyperref}
\usepackage[utf8]{inputenc}
\usepackage{multirow}
\usepackage{mathtools}
\usepackage{ulem}
\usepackage{setspace}
\usepackage{pdflscape}
\usepackage[color,notref,notcite]{showkeys}
\usepackage{enumerate}
\definecolor{labelkey}{rgb}{0.6,0,1}

\theoremstyle{plain}
\newtheorem{thm}{Theorem}
\newtheorem{prop}{Proposition}

\newtheorem{corollary}{Corollary}
\newtheorem{ass}{Assumption}

\theoremstyle{definition}
\newtheorem{definition}{Definition}
\newtheorem{remark}{Remark}
\newtheorem{example}{Example}

\def\disp{\displaystyle}

\newcommand{\dB}{\ensuremath{\mathbb{B}}}

\newcommand{\dF}{\ensuremath{\mathbb{F}}}

\newcommand{\dH}{\ensuremath{\mathbb{H}}}
\newcommand{\I}{\ensuremath{\mathbb{I}}}

\newcommand{\dN}{\ensuremath{\mathbb{N}}}

\newcommand{\dR}{\ensuremath{\mathbb{R}}}

\newcommand{\dT}{\ensuremath{\mathbb{T}}}

\newcommand{\cA}{\ensuremath{\mathcal{A}}}
\newcommand{\cB}{\ensuremath{\mathcal{B}}}

\newcommand{\cD}{\ensuremath{\mathcal{D}}}
\newcommand{\cE}{\ensuremath{\mathcal{E}}}
\newcommand{\cF}{\ensuremath{\mathcal{F}}}
\newcommand{\cG}{\ensuremath{\mathcal{G}}}

\newcommand{\cI}{\ensuremath{\mathcal{I}}}
\newcommand{\cJ}{\ensuremath{\mathcal{J}}}
\newcommand{\cK}{\ensuremath{\mathcal{K}}}
\newcommand{\cL}{\ensuremath{\mathcal{L}}}

\newcommand{\cN}{\ensuremath{\mathcal{N}}}

\newcommand{\cP}{\ensuremath{\mathcal{P}}}
\newcommand{\cQ}{\ensuremath{\mathcal{Q}}}
\newcommand{\cR}{\ensuremath{\mathcal{R}}}
\newcommand{\cS}{\ensuremath{\mathcal{S}}}

\newcommand{\cU}{\ensuremath{\mathcal{U}}}

\newcommand{\cW}{\ensuremath{\mathcal{W}}}

\def\i1{ [-\infty,\infty]}

\def\sn{n^{1/2} }
\def\sni{n^{-1/2} }

\def\ba{\mathbf{a}}
\def\bb{\mathbf{b}}

\def\bt{\mathbf{t}}
\def\bu{\mathbf{u}}
\def\bv{\mathbf{v}}
\def\bx{\mathbf{x}}

\def\bA{\mathbf{A}}

\def\bF{\mathbf{F}}
\def\bG{\mathbf{G}}
\def\bH{\mathbf{H}}

\def\bJ{\mathbf{J}}
\def\bL{\mathbf{L}}

\def\bT{\mathbf{T}}
\def\bU{\mathbf{U}}

\def\bX{\mathbf{X}}

\def\bbeta{{\boldsymbol\eta}}
\def\bgamma{{\boldsymbol\gamma}}
\def\bGamma{{\boldsymbol\Gamma}}

\def\bmu{{\boldsymbol\mu}}

\def\bvarphi{{\boldsymbol\varphi}}
\def\bpsi{{\boldsymbol\psi}}
\def\brho{{\boldsymbol\rho}}

\def\btheta{{\boldsymbol\theta}}
\def\bzeta{{\boldsymbol\zeta}}

\def\bTheta{{\boldsymbol\Theta}}

\def\setd{\{1,\ldots,d\}}
\def\setn{\{1,\ldots,n\}}

\def\setp{\{1,\ldots,p\}}

\begin{document}

\title[Semiparametric inference for random vectors]{Identifiability and inference for copula-based semiparametric models for random vectors with arbitrary marginal distributions}

\author{Bouchra R. Nasri}
\address{Département de médecine sociale et préventive, École de santé publique, Université de Montréal, C.P. 6128, succursale Centre-ville
Montréal (Québec)  H3C 3J7 }
\email{bouchra.nasri@umontreal.ca}

\author{Bruno N. R\'{e}millard}
\address{GERAD and Department of Decision Sciences, HEC Montr\'{e}al\\
3000, che\-min de la C\^{o}\-te-Sain\-te-Ca\-the\-ri\-ne,
Montr\'{e}al (Qu\'{e}\-bec), Canada H3T 2A7}
\email{bruno.remillard@hec.ca}

\thanks{Funding in partial support of this work was provided by the Fonds
qu\'e\-b\'e\-cois de la re\-cher\-che en sant\'e and the Natural Sciences and Engineering Research Council of Canada.}


\begin{abstract}
In this paper, we study the identifiability and the estimation of the parameters of a copula-based multivariate model when the margins are unknown and are arbitrary, meaning that they can be  continuous, discrete, or mixtures of continuous and discrete. When at least one margin is not continuous,  the  range of values determining the copula is not the entire unit square and this  situation  could  lead to identifiability issues that are discussed here. Next, we propose estimation methods when the margins are unknown and arbitrary, using pseudo log-likelihood adapted to the case of discontinuities. In view of applications to large data sets, we also propose a pairwise composite pseudo log-likelihood.  These methodologies can also be easily modified to cover the case of parametric margins.  One of the main theoretical result is an extension to arbitrary distributions of known convergence results of rank-based statistics when the margins are continuous. As a by-product,  under smoothness assumptions, we obtain that the asymptotic distribution of the estimation errors of our estimators are  Gaussian.   Finally, numerical experiments are presented to assess the finite sample performance of the estimators, and the usefulness of the proposed methodologies is illustrated with a copula-based regression model  for hydrological data. The proposed estimation is implemented in the R package CopulaInference \citep{Nasri/Remillard:2023c}, together with a function for checking identifiability.
\end{abstract}


\maketitle

\section{Introduction}
Copula-based models have been widely used in many applied fields such as finance \citep{Embrechts/McNeil/Straumann:2002,Nasri/Remillard/Thioub:2020}, hydrology \citep{Genest/Favre:2007,Zhang/Singh:2019} and medicine \citep{Clayton:1978,deLeon/Wu:2011}, to cite a  few. According to Sklar's representation \citep{Sklar:1959}, for a random vector $\bX = (X_1,\ldots,X_d)$ with joint distribution function $H$ and margins $F_1, \ldots, F_d$, there exists a non-empty set $\cS_H$ of copula functions $C$ so that for any $x_1,\ldots, x_d$,
$$
H(x_1,\ldots,x_d) = C\{F_1(x_1),\ldots,F_d(x_d)\},\quad \text{ for any } C\in \cS_H.
$$
If all margins are continuous, then $\cS_H$ contains a unique copula which is the distribution function of the random vector $\bU = (U_1,\ldots,U_d)$, with $\bU_j = F_j(X_j)$, $j\in \setd$. However, in many applications, discontinuities  are often present in one or several margins. Whenever at least one margin is discontinuous, $\cS_H$ is infinite, and  in this case,
any  $C  \in\cS_H$  is only uniquely  defined  on the closure  of the range $\cR_\bF= \cR_{\cF_1} \times \cdots \cR_{\cF_d}$  of $\bF = (\cF_1,\ldots,\cF_d)$,  where
$\cR_{\cF_j}=\{F_j(y): \; y\in\dR\}$ is the range of $F_j$, $j\in\setd$. This  can  lead to identifiability issues raised  in
\cite{Faugeras:2017, Geenens:2020,Geenens:2021}, creating also estimation problems  that need to be addressed.
However,  even if the copula is not unique,
 it still makes sense to use a   copula family $\{C_\btheta: \btheta\in \cP\}$, to define multivariate models, provided one is aware of the possible limitations. Indeed,  the copula-based model $\cK_\btheta(\bx) = C_\btheta\{F_1(x_1),\ldots, F_d(x_d)\}$, $\btheta\in\cP$, is a well-defined family of distributions for which estimating $\btheta$ is a challenge.   \\

In the literature, the case of discontinuous margins is not always treated properly. It is either ignored or, in some cases, continuous margins are  fitted to the data \citep{Chen/Singh/Guo/Mishra/Guo:2013}. This procedure does not solve the problem since there still will be ties.
An explicit  example that underlines the problem with ignoring ties is given in Section \ref{sec:est}  and Remark \ref{rem:incorrect_method}. A solution proposed in the literature is to use jittering, where data are perturbed by independent random variables, introducing extra variability. 
Our first aim is to address  the  identifiability issues for the model
$\{\cK_\btheta:\btheta\in\cP\}$, when the margins are unknown and arbitrary. Our second aim  is to present formal inference methods for the  estimation of $\btheta$. 
More precisely, we consider a semiparametric approach for the  estimation of $\btheta$, when the margins are arbitrary, i.e. each margin  can be continuous, discrete, or even a  mixture of a discrete and continuous distribution.  The estimation approach is based on pseudo log-likelihood taking into account the discontinuities. We also propose a pairwise composite log-likelihood. In the literature, few  articles focused on the estimation of the copula parameters in the case of noncontinuous margins \citep{Song/Li/Yuan:2009, deLeon/Wu:2011, Zilko/Kurowicka:2016,Ery:2016,Li/Li/Qin/Yan:2020}. Most of them have considered only the case when the components of the copula are discrete or continuous and a full parametric model, except \cite{Ery:2016} and \cite{Li/Li/Qin/Yan:2020}. In \cite{Ery:2016}, in the bivariate discrete case, the author considered a semiparametric approach and studied the asymptotic properties. In \cite{Li/Li/Qin/Yan:2020}, the authors have proposed a semiparametric approach for the estimation of the copula parameters for  bivariate data having arbitrary distributions, without presenting any asymptotic results, neither discussing identifiability issues. \\ 



The article is organized as follows: In Section \ref{sec:limits}, we discuss the important topic of identifiability,  as well as the limitations of the copula-based approach for multivariate data with arbitrary margins. Conditions are stated in order to have identifiability so that the estimation problem is well-posed. Next, in Section \ref{sec:est}, we describe the estimation methodology, using the  pseudo log-likelihood approach as well as the composite pairwise approach, while  the estimation error is studied in Section \ref{ssec:conv1}.
By using an extension of the asymptotic behavior of rank-based statistics \citep{Ruymgaart/Shorack/vanZwet:1972} to data with arbitrary distributions (Theorem \ref{thm:rum_ext}), we show that the limiting distribution of the pseudo log-likelihood  estimator is  Gaussian. Similar results are  obtained for pairwise composite pseudo log-likelihood. In addition, we can obtain similar asymptotic results if the margins are estimated from parametric families, instead of using the empirical margins, i.e., we consider the multivariate model $\cK_{\btheta,\bgamma_1,\ldots,\bgamma_d}(x_1,\ldots, x_d) = C_\btheta\{F_{1,\bgamma_1}(x_1),\ldots, F_{d,\bgamma_d}(x_d)\}$, where the margins are estimated first, and then the copula parameter $\btheta$.  This is discussed in Remark \ref{rem:IFM}.
Finally, in Section \ref{sec:exp}, numerical experiments are performed to assess the convergence and precision of the proposed estimators in bivariate and trivariate settings, while in Section \ref{sec:example}, as an example of application,  a copula-based regression approach is proposed to investigate  the relationship between the duration and severity, using hydrological data \citep{Shiau:2006}.

\section{Identifiability and limitations}\label{sec:limits}

As exemplified in \cite{Faugeras:2017, Geenens:2020}, there are several problems and limitations when applying copula-based models to data with arbitrary distributions, the most important being identifiability. This is discussed first. Then, we will examine some limitations of these models.

\subsection{Identifiability}\label{ssec:identify}

For the discussion about identifiability, we consider two cases: copula family identifiability, and copula parameter identifiability.
For copula family identifiability, it may happen that two  copulas from different families, say  $C_\btheta$ and $D_\bpsi$, belong to $\cS_H$.
For example, in \cite{Faugeras:2017}, the author considered the bivariate Bernoulli case, whose distribution is completely determined by $h(0,0)=P(X_1=0,X_2=0)$, $p_1=P(X_1=0)$, and $p_2=P(X_2=0)$. 
Provided that the copula families $C_\theta$ and $D_\psi$ are rich enough, there exist unique parameters $\theta_0$ and $\psi_0$ so that
$h(0,0) = C_{\theta_0}(p_1,p_2) = D_{\psi_0}(p_1,p_2)$. For calculation purposes, the choice of the copula in this case is immaterial,  and there is no possible interpretation of the copula or its parameter or the type of dependence induced by each copula. All that matters is $h(0,0)$, or the associated odd ratio \citep{Geenens:2020}, or Kendall's tau of $C^\maltese$, given in this case by $\tau^\maltese = 2\{h(0,0)-p_{1}p_{2}\}$. Here, $C^\maltese$ is the so-called multilinear copula, depending on the margins and belonging to $\cS_H$; see, e.g., \cite{Genest/Neslehova/Remillard:2017}. To  illustrate  the fact that   the  computations  are  the same for any $C\in\cS_H$,  consider  the conditional distribution $P(X_2\le x_2|X_1=x_1)$ of $X_2$ given $X_1=x_1$, given  by
$ \left. \dfrac{ \partial} {\partial_{u}}   C(u,v)\right|_{u=F_1(x_1), v = F_2(x_2)}$,  if $F_1$ is continuous at $x_1$ and  $
 \dfrac{C\{F_1(x_1), F_2(x_2)\}-C\{F_1(x_1-),F_2(x_2)\}}{F_1(x_1)-F_1(x_1-)}$, if $F_1$ is not continuous at $x_1$,
where $F_1(x_1-) = P(X_1< x_1)$.
The value of $P(X_2\le x_2|X_1=x_1)$ is  independent  of  the choice  of   $C\in\cS_H$, since all copulas in $\cS_H$ have the same value on the closure of the range $\cR_\bF$.
So apart from the lack of interpretation of the type of dependence, there is no problem for calculations, as long as for the chosen copula family, its parameter is identifiable, as defined next.\\

We now consider the identifiability of the copula parameter for a given copula family $\{C_\btheta : \btheta\in\cP\}$.
Since one of the aims is to estimate the parameter of the family $\{\cK_\btheta=C_\btheta\circ \bF: \btheta\in\cP\}$, the following definition of identifiability is needed.

 \begin{definition}\label{def:identity0}
 For a copula family $\{C_\btheta: \btheta\in\cP\}$ and a vector of margins $\bF = (F_1,\ldots,F_d)$, the parameter $\btheta$ is identifiable with respect to  $\bF$
 if the mapping  $\btheta\mapsto \cK_\btheta = C_\btheta\circ \bF $ is injective, i.e., if for $\btheta_1, \btheta_2 \in \cP$, $\btheta_1\neq \btheta_2$ implies that  $\cK_{\btheta_1} \neq \cK_{\btheta_2}$. This is equivalent to the existence of $\bu\in  \cR_\bF^*$ such that $C_{\btheta_1}(\bu) \neq C_{\btheta_2} (\bu)$, where for a vector  of margins $\bG$,
$$
\cR_{\bG}^* =  \left\{\bu \in  \cR_{\bG} \cap (0,1]^d : u_j < 1 \text{ for at least two indices  } j\right\}.
$$
 \end{definition}
The following result, proven in the  Supplementary Material, 
is essential for checking that a  given mapping is injective whenever $\cR_\bF^*$ is finite.
\begin{prop}\label{prop:rolle}
 Suppose $\bT$ is a continuously differentiable mapping from a convex open set $O\subset \dR^p$ to $\dR^q$. If $p>q$, then $\bT$ is not injective. Also, $\bT$ is injective if the rank of the derivative $\bT'$ is $p\le q$.
Furthermore, if the rank of $\bT'(\btheta_0)$ is $p$ for  some $\btheta_0 \in O$, the rank of $\bT'(\btheta)$ is $p$ for any $\btheta$ in some neighborhood of $\btheta_0$. Finally, if the maximal rank is $r<p$, attained at some $\btheta_0\in O$, then  the rank of $\bT'$ is $r$ in some  neighborhood of $\btheta_0$, and $T$ is not injective.
\end{prop}

\begin{example}\label{ex:monotone}
A 1-dimensional parameter is identifiable for any $\bF$ for the following  copula families  and their rotations: the bivariate Gaussian copula, the Plackett copula, as well as the multivariate Archimedean copulas with one parameter like the Clayton, Frank, Gumbel, Joe.  This is because for any fixed $\bu\in (0,1)^d$, $\btheta\mapsto C_\btheta(\bu)$ is strictly monotonic.  Note also that by \cite[Theorem A2]{Joe:1990}, every meta-elliptic copula $C_\brho$, with $\brho$ the correlation matrix parameter, is ordered  with respect to $\rho_{ij}$.
\end{example}

In practice, for copula models with bivariate or multivariate parameters, it might be impossibly difficult to verify injectivity since the margins are never known. In this case,  since $\bF_n = (F_{n1},\ldots,F_{nd})$ converges to $\bF$,
 where $\disp F_{nj}(y) = \dfrac{1}{n+1}\sum_{i=1}^n \I(X_{ij}\le y)$, $y\in \dR$, $j\in \setd$,
 one could verify that the parameter is identifiable with respect to $\bF_n$, i.e.,  that the mapping $\btheta\mapsto C_\btheta\circ \bF_n$ is injective. The latter does not necessarily imply identifiability with respect to $\bF$ but if the parameter is not identifiable with respect to $\bF_n$, one should choose another parametric family $C_\btheta$.
 Next, the cardinality of  $\cR_{\bF_n}^*$ is
 $q_n= m_{n1} \times  \cdots \times m_{nd} - \sum_{j=1}^d m_{nj} +d-1$, where $m_{nj}$ is the size of the support of $F_{nj}$, $j\in\setd$. Let $ \{\bu_i: 1\le i\le q_n\}$ be an enumeration of $\cR_{\bF_n}^* $.  As a result,
 $\btheta\mapsto C_\btheta\circ \bF_n$ is injective iff the mapping $\btheta\mapsto T_n(\btheta) = \{ C_\btheta(\bu_1), \ldots, C_\btheta(\bu_{q_n})\}^\top$ is injective,
 There are two cases to consider: $p>q_n$ or $p\le q_n$.
 \begin{itemize}
     \item First, one should not choose a copula family with $p>q_n$, since in this case, according to Proposition
     \ref{prop:rolle}, the mapping $\bT_n$ cannot be injective, so the parameter is not identifiable.\\

     \item  Next, if $p\le q_n$, it follows from Proposition \ref{prop:rolle}, 
      that $T_n$ is injective if $\bT_n'$ has rank $p$. Also, it follows from Proposition \ref{prop:rolle}
     that if the rank of $\bT_n'(\btheta)$ is $p$, then there is a neighborhood $\cN \subset \cP$ of $\btheta$ for which the rank of $\bT_n'(\tilde\btheta)$ is $p$, for any $\tilde\btheta\in\cN$. One could then restrict the estimation to this neighborhood if necessary.  Note that the matrix $\bT_n'(\btheta)$ is composed of the (column)  gradient vectors $\dot C_\btheta(u_j)$, $j \in \{1,\ldots, q_n\}$. The latter can be calculated explicitly or it can be approximated numerically. To check that the rank of $T'$ is $p$,
     one needs to choose a bounded neighborhood $\cN$ for $\btheta$, choose an appropriate covering of $\cN$ by balls of radius $\delta$, with respect to an appropriate metric, corresponding to a given accuracy $\delta$ needed for the estimation,  and then compute  the rank of the mappings $T_n'(\btheta)$, for all the centers $\btheta$ of the covering balls. This appears in the R package CopulaInference \citep{Nasri/Remillard:2023c}.
 \end{itemize}


\begin{example}\label{ex:bernoulli} In the bivariate Bernoulli case, $q=q_n=1$, so a copula parameter is identifiable  if for any fixed $\bu\in (0,1)^2$, the mapping $\btheta\mapsto C_\btheta(\bu)$ is injective.
As a result of the previous discussion, $p=1$, i.e., the parameter should be 1-dimensional and for any fixed $\bu\in (0,1)^2$, the mapping $\theta\mapsto C_\theta(\bu)$ must be  strictly monotonic since $\partial_\theta C_\theta(\bu)$ must be always positive or always negative.
\end{example}

\begin{remark}[Student  copula]\label{rem:student}
There are cases when $p<q_n$ is needed. For example, for the bivariate Student copula, $q_n \ge 3>p=2$ is sometimes necessary. To see why,
note that the point $\bu = (1/2,1/2)$ determines $\rho$ uniquely since for any $\nu>0$, $C_{\rho,\nu}\left(\dfrac{1}{2},\dfrac{1}{2}\right)= \dfrac{1}{4}+\dfrac{1}{2\pi}\arcsin{\rho}$. Take for example $\rho=0.3$. Then $C_{0.3,\nu}(0.75,0.55) = 0.452$ is not injective, having $2$ solutions, namely $\nu=0.224965$ and $\nu = 0.79944$,  so the rank of $T'(\rho,\nu)$ is
 $1 < p$ at points $\bu_1=(0.5,0.5)$ and $\bu_2 = (0.75,0.55)$.
 However, if $\nu$ is restricted to  a finite set like $\{k:\; 1\le k \le 50\} $, then the mapping is injective. This restriction of the values of $\nu$ makes sense if one wants to use the  Student copula. In fact,  since  exact values of $C_{\rho,\nu}$ can be computed explicitly only when $\nu$ is an integer \citep{Dunnett/Sobel:1954}. It is the case for example in the R package \textit{copula}. Otherwise, one needs to use numerical integration \citep{Genz:2004}, with might lead to numerical inconsistencies while differentiating the copula with respect to $\nu$ in a given open interval.
 \end{remark}

\subsection{Limitations}\label{ssec:limits}

Having  discussed identifiability issues, we are now in a position to discuss limitations of the copula-based approach for modeling multivariate data. Two main issues have been  identified when one or some margins have discontinuities: interpretation and dependence on margins.

\subsubsection{Interpretation}
As exemplified in the extreme case of the bivariate Bernoulli distribution discussed in Example \ref{ex:bernoulli}, interpretation of the copula or its parameter can be hopeless.
Recall that the object of interest is the family $\{\cK_\btheta : \btheta\in\cP\}$, not the copula family. Even if $X_1$ has only one atom at $0$ and $X_2$ is continuous, then one  only ``knows'' $ C_\btheta$ on $\bar \cR= \{[0,F_1(0-)]\cup [F_1(0),1]\}\times [0,1]$. For example, how can we interpret the form of dependence induced by a Gaussian copula or a Student copula restricted to such $\bar\cR$? Apart from tail dependence, there is not much one can say.
As it is done in the case for bivariate copula, one could plot the graph of a large sample, say $n=1000$. The scatter plot is not very informative. This is even worse if the support of $H$ is finite. As an illustration, we generated 1000 pairs of observations from Gaussian, Student(2.5), Clayton, and Gumbel copulas with $\tau=0.7$, where the margin $F_1$ is a mixture of a Gaussian  and an atom at $0$ (with probability $0.1$), and $F_2$ is Gaussian. The scatter plots are given in Figure \ref{fig:simgaussian}, together with the scatter plots of the associated pseudo-observations $\left( F_{n1}(X_{i1}),  F_{n2}(X_{i2}) \right)$.  The raw datasets in panel a) of Figure \ref{fig:simgaussian} do not say much, apart from the fact that $X_1$ has an atom at $0$, while pabel b) of Figure \ref{fig:simgaussian} illustrates the fact that the copula is only known on $\bar \cR$. All the graphs are similar but for the Clayton copula. 
So one can see that even if  there is only one atom, one cannot really interpret the type of dependence.

\begin{figure}[ht!]
    \centering
a) \includegraphics[scale=0.15]{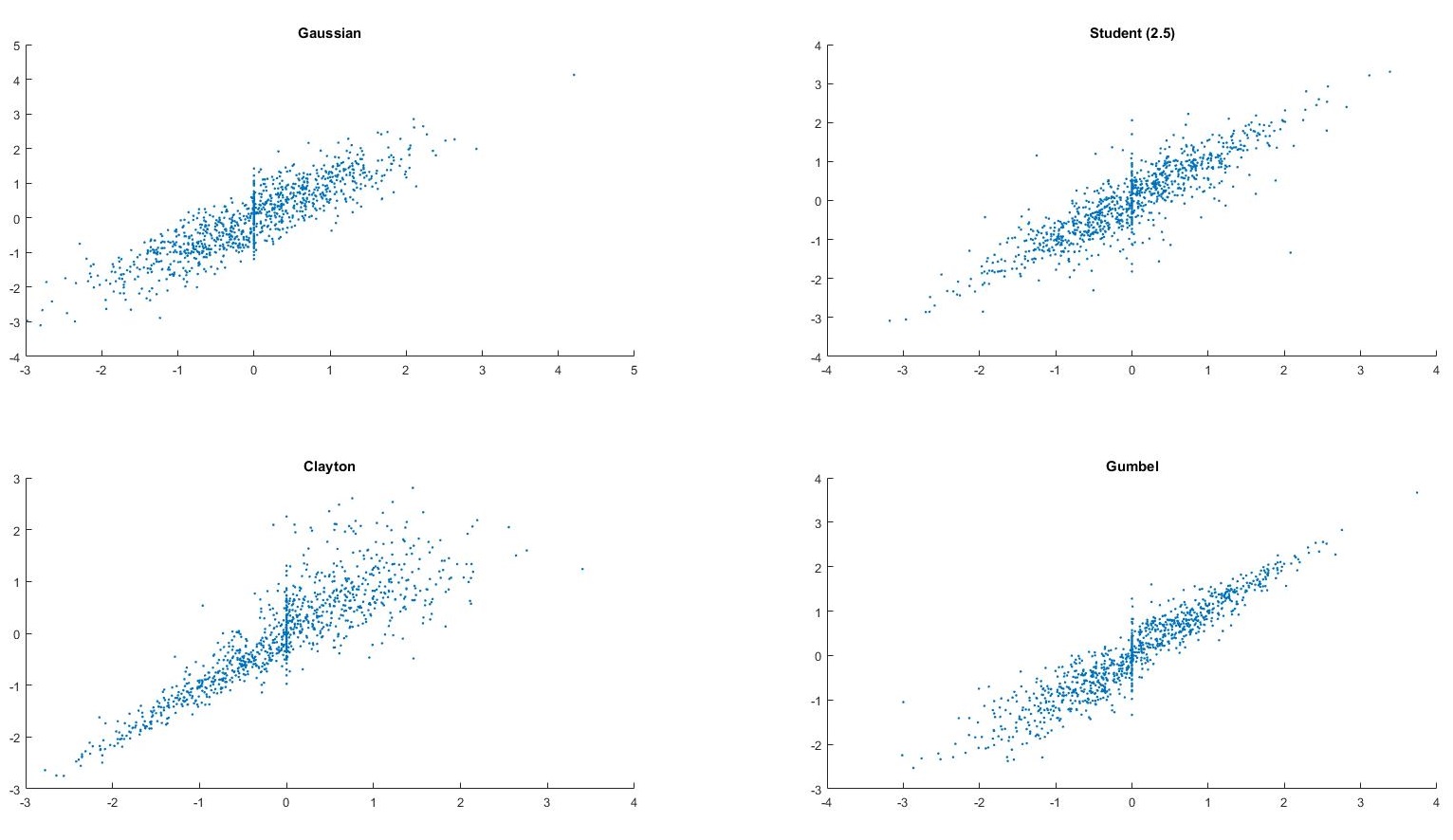}
b)  \includegraphics[scale=0.15]{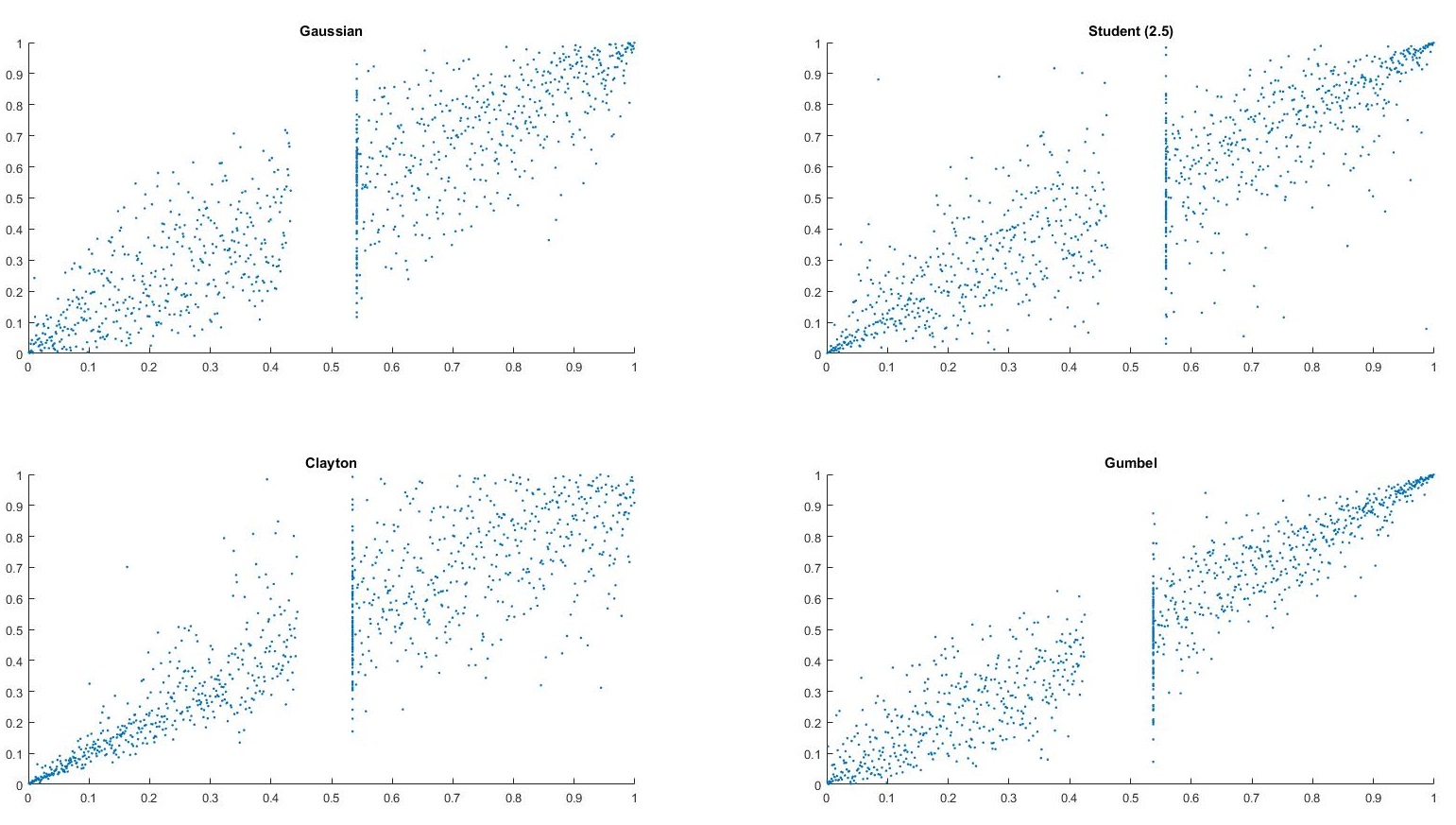}
    \caption{Panel a): scatter plots of 1000 pairs from Gaussian, Student(2.5), Clayton, and Gumbel copulas with $\tau=0.7$, where $F_1$ is a mixture of a Gaussian (with probability $0.9$) and a Dirac at $0$, and $F_2$ is Gaussian. Panel b): scatter plots of the associated pseudo-observations.}
    \label{fig:simgaussian}
\end{figure}

\subsubsection{Dependence on margins}\label{sssec:margins}

In the bivariate Bernoulli case, \cite{Geenens:2020} proposed the odds-ratio $\omega$ as  a ``margin-free'' measure of dependence. In fact, \cite{Geenens:2020}   quotes Item 3 of  \cite[Theorem 6.3 (p. 116)]{Rudas:2018} which says that if $\btheta_1 = (p_1,p_2)$ are the margins' parameters and  $\theta$ is a parameter, whose range does not depend on $\btheta_1$ (called variation independent parametrization in \cite{Rudas:2018}), and if $(\btheta_1,\theta)$  determines the full distribution $H$,  then $\theta$ is a one-to-one function of the odd-ratio $\omega$. However, as a proof of this statement, \cite{Rudas:2018} says that it is because $(\btheta_1,\omega)$ determines $H$. This only proves that $\theta$ is a one-to-one function of $\omega$, for a fixed $\btheta_1$, i.e., for fixed margins, not that it is a one-to-one function of $\omega$ alone. In fact, according to Rudas' definition,  for the Gaussian copula $C_\rho$, $\rho\in [-1,1]$, it follows that $(\btheta_1,\rho)$  is also a variation independent parametrization of $H$, by simply setting $H(p_1,p_2) =C_\rho(p_1,p_2)$.  Therefore, $\rho$ qualifies as a ``margin-free'' association measure as well, if one agrees that margin-free means ``variation independent in the sense of \cite{Rudas:2018}''. Using the same argument as  in \cite{Rudas:2018}, it follows that any variation independent parameter is a function of $\rho$.
This construction works for any rich enough copula family such as Clayton, Frank, or Plackett (yielding the omega ratio), and that it also applies to all margins, not only Bernoulli margins. In fact, the property of being one-to-one is the same as what we defined as ``parameter identifiability'' in Section \ref{ssec:identify}.

\section{Estimation of parameters}\label{sec:est}

In  this section, we will show how to estimate the parameter $\btheta$ of the semiparametric model $\{\cK_\btheta = C_\btheta\circ F: \btheta\in \cP\}$, where $\cP\subset \dR^p$ is convex.  It is assumed that the  given copula family $\{ C_\btheta: \btheta\in \cP\}$ satisfies the following assumption.

\begin{ass}\label{hyp:cop}
$\btheta \mapsto C_\btheta $ on $\cP$ is thrice continuously differentiable with respect to $\btheta$, for any $\btheta\in \cP$,  the density $c_\btheta$ exists, is thrice continuously differentiable, and is strictly positive on  $(0,1)^d$.
Furthermore, for a given vector of margins $\bF$, $\btheta \mapsto C_\btheta\circ \bF $ is injective on $\cP$.
\end{ass}
Suppose that $\bX_1,\ldots,\bX_n$ are iid with $\bX_i\sim \cK_{\btheta_0}$, for some $\btheta_0\in\cP$. Without loss of generality, we may suppose that $\bX_i = \bF^{-1}(\bU_i)$, where
$\bU_1,\ldots,\bU_n$ are iid with $\bU_i\sim C_{\btheta_0}$.
Estimating parameters for arbitrary distributions is more challenging that in the continuous case.
For continuous (unknown) margins, copula parameters can be estimated using different approaches, the most popular ones based on pseudo log-likelihood being the normalized ranks method \citep{Genest/Ghoudi/Rivest:1995,Shih/Louis:1995} and the IFM method (Inference Functions for Margin) \citep{Joe/Xu:1996, Joe:2005}.
Since the margins are unknown, it is tempting to ignore that some margins are discontinuous and consider maximizing the (averaged)  pseudo log-likelihood  $
L_n(\btheta) = \dfrac{1}{n}\sum_{i=1}^n \log{ c_\btheta(\bU_{n,i})}$, $\bU_{n,i}=\bF_{n}(\bX_{i})$, $i\in\setn$.
Note than one could also replace the nonparametric margins with parametric ones, as in the IFM approach.
However, in either cases, there is a problem with using $L_n$ when there are atoms.  To see this, consider a simple  bivariate model with Bernoulli margins, i.e., $P(X_{ij}=0)= p_j$, $P(X_{ij}=1)=1-p_j$, $j\in \{1,2\}$. Then, the full (averaged)  log-likelihood is
 \begin{eqnarray*}
\ell_n^\star (\btheta,p_1,p_2) &=&  h_n(0,0)\log\{C_\btheta(p_{1},p_{2})\}+
 h_n(0,1)\log\{p_{1} - C_\btheta(p_{1},p_{2})\}\\
&&
+ h_n(1,0)\log\{p_{2} - C_\btheta(p_{1},p_{2})\} + h_n(1,1)\log\{1-p_{1} -p_{2}+ C_\btheta(p_{1},p_{2})\},
\end{eqnarray*}
where $h_n(x_1,x_2) = \dfrac{1}{n}\sum_{i=1}^n
\I\{X_{i1}=x_1,X_{i2}=x_2\}$. Hence, the pseudo log-likelihood $\ell_n$ for $\btheta$, when  $p_1$ and $p_2$ are estimated by
$p_{n1} = h_n(0,0)+h_n(0,1)$ and $p_{n2} = h_n(0,0)+h_n(1,0)$,  is
\begin{eqnarray*}
{\ell}_n  (\btheta)&=&  h_n(0,0)\log\{C_\btheta(p_{n1},p_{n2})\}+
 h_n(0,1)\log\{p_{n1} - C_\btheta(p_{n1},p_{n2})\}\\
&&
+ h_n(1,0)\log\{p_{n2} - C_\btheta(p_{n1},p_{n2})\} + h_n(1,1)\log\{1-p_{n1} -p_{n2}+ C_\btheta(p_{n1},p_{n2})\}.
\end{eqnarray*}
It is clear that under general conditions, when the parameter is 1-dimensional,
 and the copula family is rich enough,
maximizing $\ell_n$ produces a consistent estimator for $\btheta$ while maximizing $L_n$ will not.
In fact, for $\ell_n$, the estimator $\theta_n$ of $\theta$ is the solution of $C_\theta(p_{n1},p_{n2}) = h_n(0,0)$ \citep[Example 13]{Genest/Neslehova:2007}.
Note also that $\ell_n$ converges almost surely to
\begin{eqnarray*}
\ell_\infty(\btheta) &=& h(0,0)\log\{C_\btheta(p_{1},p_{2})\}+
 h(0,1)\log\{p_{1} - C_\btheta(p_{1},p_{2})\}\\
&& \qquad + h(1,0)\log\{p_{2} - C_\btheta(p_{1},p_{2})\}
 + h(1,1)\log\{1-p_{1} -p_{2}+ C_\btheta(p_{1},p_{2})\},
\end{eqnarray*}
where $h(i,j)=P(X_1=i,X_2=j)$, $i,j\in \{0,1\}$.
However, the estimator obtained by maximizing $L_n$ is not consistent in general.
In fact, for any copula $C_\theta$ with a density $c_\theta$ which is continuous and non-vanishing on $(0,1]^2$,
 \begin{eqnarray*}
L_n (\btheta) &= &  h_n(0,0)\log{ c_\theta\left(\dfrac{n}{n+1}p_{n,1}, \dfrac{n}{n+1}p_{n,2}\right)}+
 h_n(1,0)\log{c_\theta\left(\dfrac{n}{n+1}, \dfrac{n}{n+1}p_{n,2}\right)}\\
 &&\quad+ h_n(0,1)\log{c_\theta\left(\dfrac{n}{n+1}p_{n,1}, \dfrac{n}{n+1}\right)}+ h_n(1,1)\log{c_\theta\left(\dfrac{n}{n+1}, \dfrac{n}{n+1}\right)}
  \end{eqnarray*}
converges almost surely to $L_\infty (\btheta) = h(0,0)\log{c_\theta\left(p_1, p_2\right)}+
 h(1,0)\log{c_\theta\left(1, p_2\right)}+ h(0,1)\log{c_\theta\left(p_1, 1 \right)}+ h(1,1)\log{c_\theta\left( 1,1\right)}.$
In particular, for a Clayton copula with $\theta_0=2$, one gets $L_\infty (\btheta) = \log(1+\theta)+ \theta ( p_1\log{p_1} +p_2\log{p_2}) +h(0,0)(1+2\theta)
 \left\{\log{C_\theta(p_1,p_2)} -\log(p_1 p_2)\right\}$,
with $h(0,0) = C_{\theta_0}(1/2,1/2) = 1/\sqrt{7}$. As displayed in Figure \ref{fig:claytonbernoulli}, for the limit $L_\infty$ of $L_n$,  the supremum is attained at $\btheta= 4.9439$, while for the limit $\ell_\infty$ of $\ell_n$, the supremum is attained at the correct value $\theta=\theta_0=2$.
\begin{figure}[ht!]
    \centering
a) \includegraphics[scale=0.2]{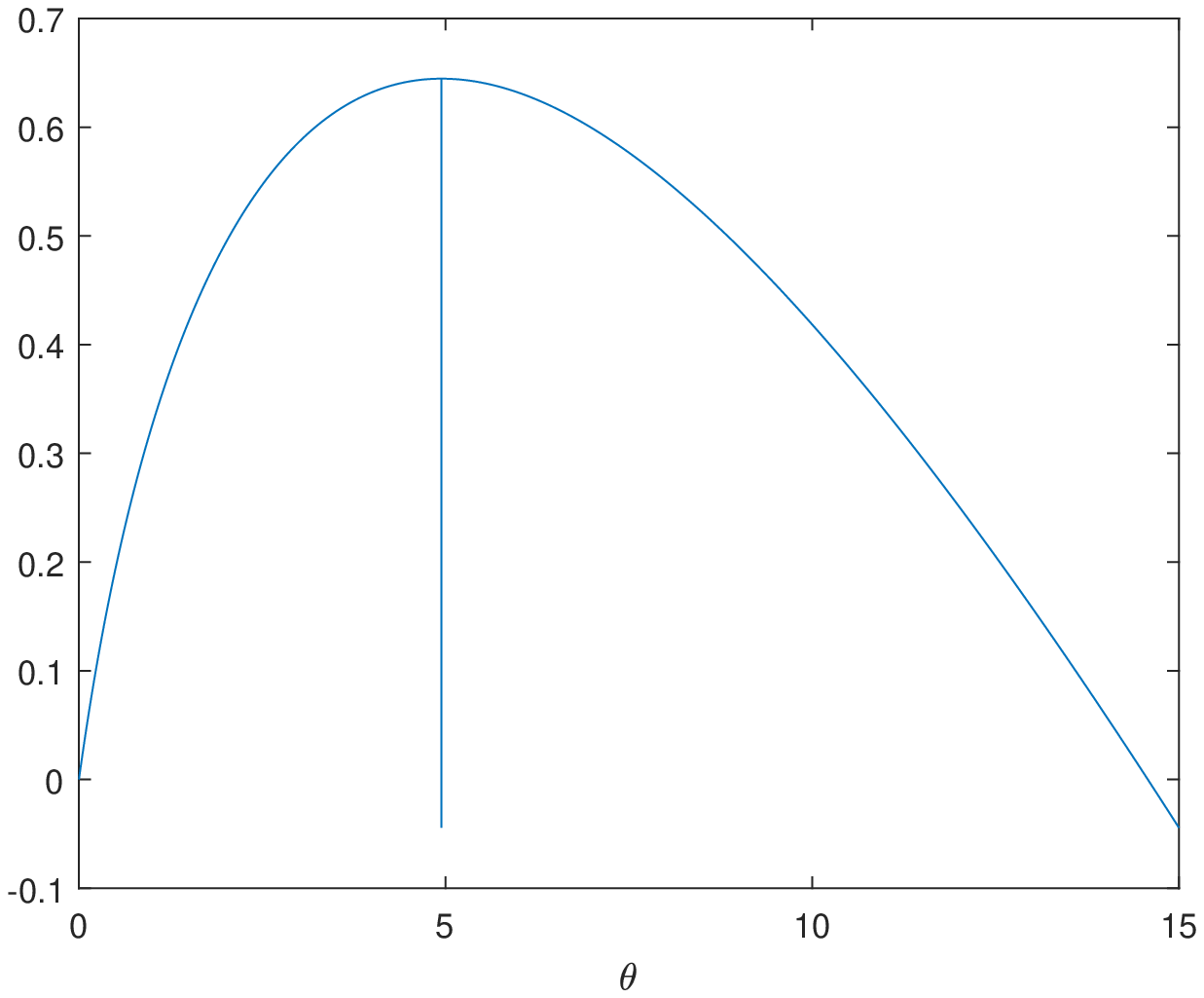} b) \includegraphics[scale=0.2]{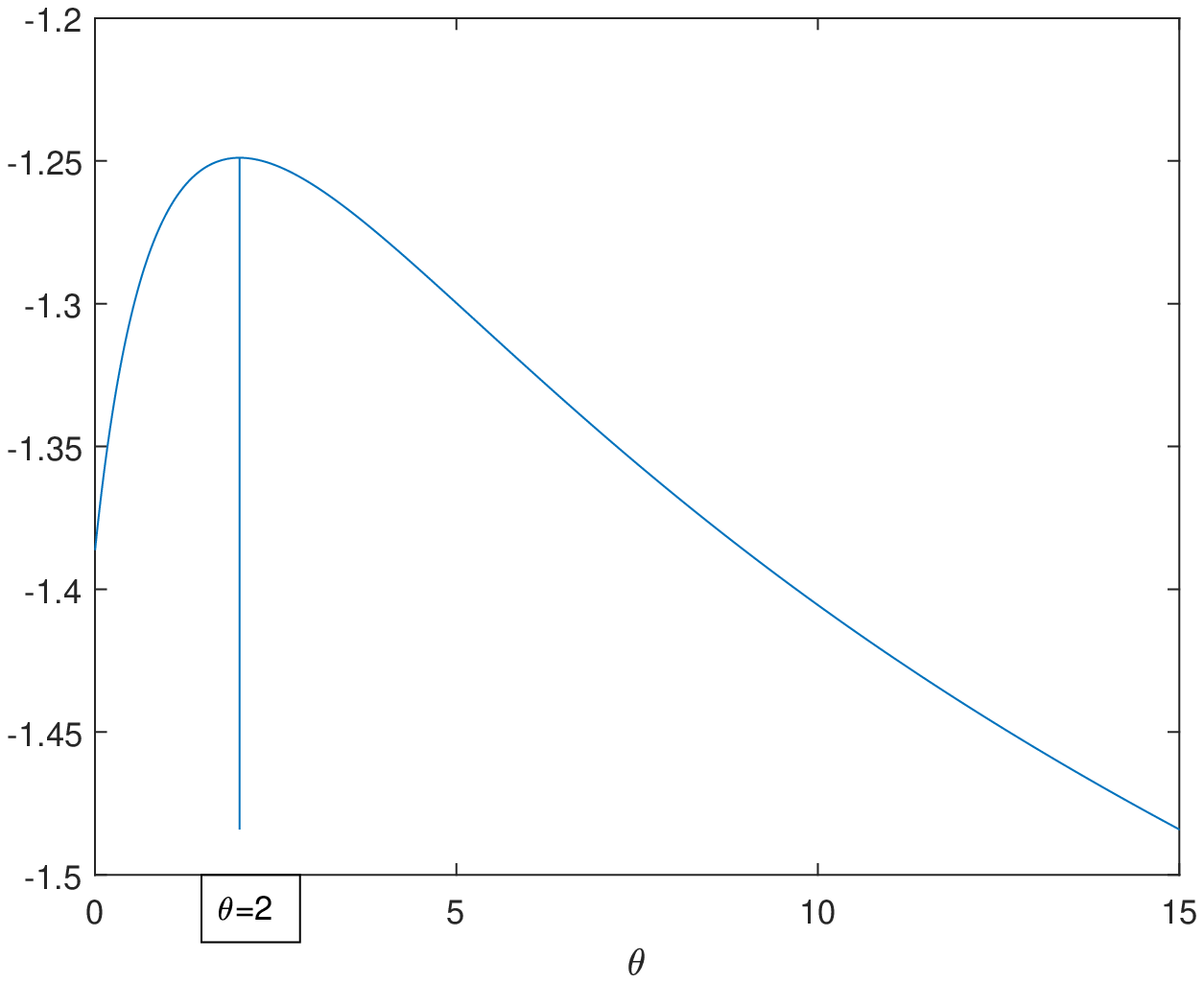}
    \caption{Graphs of $L_\infty(\theta)$ (a) and $\ell_\infty(\theta)$ (b)  for the bivariate Bernoulli case with Clayton copula ($\theta_0=2$.) }
    \label{fig:claytonbernoulli}
\end{figure}
\begin{remark}\label{rem:incorrect_method}
This simple example shows that one must be very careful in estimating copula parameters when the margins are not continuous. In particular, one should not use the usual pseudo-MLE method based on $L_n$ or the IFM method with continuous margins fitted to data with ties.
\end{remark}
\subsection{Pseudo log-likelihoods}\label{sec:lik}

For any $j\in \setd$, let $\cA_j = \{x\in \dR: \Delta F_j(x)>0\}$ be the (countable) set of atoms of $F_j$, where for a general univariate distribution function $\cG$, $\Delta \cG(x) = \cG(x) - \cG({x}-)$, with $\cG({x}-) = \lim_{n\to\infty}\cG(x-1/n)$. Throughout this paper, we assume that $\cA_j$ is closed.
For $j\in\setd$,  $\mu_{cj}$ denotes the counting measure on $\cA_j$,  and let $\cL$ be Lebesgue's measure; both measures are defined on $(\dR,\cB_\dR)$, and  $\mu_{cj}+\cL$, also defined on $(\dR,\cB_\dR)$, is $\sigma$-finite.
Further
let $\mu$ be the product measure on $\dR^d$ defined by $\mu = (\mu_{c1}+\cL)\times \cdots \times (\mu_{cd}+\cL)$, which is also $\sigma$-finite. In what follows, $
\cG^{-1}(u) = \inf\{x: \cG(x)\ge u\}$,  $u\in (0,1)$. 
Since our aim is to estimate the parameter of the copula family without knowing the margins, the following assumption is necessary.

\begin{ass}\label{hyp:margins}
The margins $F_1,\ldots, F_d$ do not depend on the parameter
$\btheta \in \cP \subset \dR^p$. In addition, for any $j\in \setd$, $F_j$ has density $f_j$ with respect to the measure $\mu_{cj}+\cL$.
\end{ass}
What is meant in Assumption \ref{hyp:margins} is that if the margins were parametric (e.g., Poisson), their parameters are not related to the parameter of the copula. This assumption  is needed when one wants to estimate the margins first, and then estimate $\btheta$.
Assumption \ref{hyp:margins} really means that $\nabla_\btheta F_j(x_j) \equiv 0$, which is a natural assumption. This assumption is also implicit in the continuous case.\\

To estimate the copula parameters,  one could use at least two different pseudo log-likelihoods: an informed one, if $\cA_1,\ldots, \cA_d$ are known, and a non-informed one,  if some atoms are not  known. The latter is the approach proposed in \cite{Li/Li/Qin/Yan:2020} in the bivariate case.  From a practical point of view, it is easier to implement, since there is no need to define the atoms. However, it requires more assumptions, and one could argue that in practice, atoms should be known.
Before writing these pseudo log-likelihoods, we need to introduce some notations.
Suppose that $C$ is a copula having a continuous density $c$ on $(0,1)^d$. For any $B\subset \setd$ and for any $\bu = (u_1,\ldots,u_d) \in (0,1)^d$, set $\partial_B C (\bu)  = \{\prod_{j\in B}\partial_{u_j} \}C(\bu)$, where $\partial_{\setd} C(\bu) = c(\bu)$ is the density of $C_\btheta$ and $\partial_\emptyset C(\bu) = C(\bu)$. Finally, for any vector $\bG = (\cG_1,\ldots,\cG_d)$ of margins, and any $B\subset \setd $, set
\begin{equation}\label{eq:FB}
\left(\tilde \bG^{(B)}(\bx)\right)_j =
\left\{
\begin{array}{ll}
\cG_{j}({x_{j}}-)  & \text{, if } j \in B;\\
 \cG_{j}(x_{j})   & \text{, if } j \in B^\complement = \setd\setminus B.
 \end{array}\right.
 \end{equation}
Under Assumptions \ref{hyp:cop}--\ref{hyp:margins}, it follows from Equation (3) 
in the Supplementary Material, 
that the full (averaged)  log-likelihood is given by
\begin{equation*}
\ell_{n}^\star(\btheta)  %
=  \dfrac{1}{n}  \sum_{A\subset \{1,\ldots,d\}}   \sum_{i=1}^n J_A(\bX_i) \log{K_{A,\btheta}(\bX_i)} + \dfrac{1}{n}  \sum_{A\subset \{1,\ldots,d\}}\sum_{i=1}^n \sum_{j \in A^\complement}\log f_j(X_{ij}),
\end{equation*}
where, for any $A\subset \setd$,  $J_A(\bx)=  \left[\prod_{j\in A} \I\{x_j\in \cA_j \}\right] \left[\prod_{j\in A^\complement} \I\{x_j\not\in \cA_j \}\right]$ and
\begin{equation}\label{eq:GAtheta}
 K_{A,\btheta} (\bx)= \sum_{B\subset A} (-1)^{| B|} \partial_{A^\complement} C_\btheta\left\{\tilde\bF^{(B)}(\bx)\right\},
\end{equation}
with the usual convention that a product over the empty set is $1$.
If the margins were known or estimated first, then maximizing
$\ell_n^\star(\btheta)$ or maximizing
$\disp
\sum_{A\subset \{1,\ldots,d\}}   \sum_{i=1}^n J_A(\bX_i) \log{K_{A,\btheta}(\bX_i)}$ with respect to $\btheta$ would be equivalent, 
since by Assumption \ref{hyp:margins}, the margins do not depend on $\btheta$.
As a result,  one gets the ``informed'' (averaged)  pseudo log-likelihood $\ell_n$, defined by
\begin{equation}\label{eq:ll2d}
\ell_{n}(\btheta)  %
=  \dfrac{1}{n} \sum_{A\subset \{1,\ldots,d\}}  \sum_{i=1}^n J_A(\bX_i) \log{K_{n,A,\btheta}(\bX_i)},
\end{equation}
where for any $A\subset \setd$,
$\disp K_{n,A,\btheta}(\bx)= \sum_{B\subset A} (-1)^{| B|} \partial_{A^\complement} C_\btheta\left\{\tilde\bF_{n}^{(B)}(\bx)\right\}$,
 and
 $ \tilde\bF_{n}^{(B)}(\bx)$ is defined by \eqref{eq:FB}, with $\bG=\bF_n$.
 Setting $\btheta_n= \disp \arg\max_{\btheta\in\cP} \ell_n(\btheta)$, then we define $ \bTheta_n = \sn(\btheta_n-\btheta_0)$.
As \cite{Li/Li/Qin/Yan:2020} did in the bivariate case, we can also define the ``non-informed'' (averaged)  pseudo log-likelihood  $\tilde\ell_{n}(\btheta)$, where 
$J_{n,A}(\bx)=\left[\prod_{j\in A} \I\{\Delta F_{nj}(x_j) > 1/(n+1) \}\right] \left[\prod_{j\in A^\complement} \I\{\Delta F_{nj}(x_j) = 1/(n+1) \}\right]$  replaces $J_A$ in \eqref{eq:ll2d}. Of course, for a given $i$, if $X_{ij}$ is an atom, then, when $n$ is large enough, $\Delta F_{nj}(X_{ij})> 1/(n+1)$. However, for a given $n$, it is possible that $\Delta F_{nj}(X_{ij})= 1/(n+1)$ even when $X_{ij}$ is an atom.
If the real value of the parameter is $\btheta_0$ and $\tilde\btheta_n= \disp \arg\max_{\btheta \in \cP} \tilde\ell_n(\btheta)$,  we set $\tilde \bTheta_n = \sn\left(\tilde \btheta_n-\btheta_0\right)$.

\begin{example}\label{ex:bivLL}
If $d=2$, then
$
K_{n,\emptyset,\btheta}(x_1,x_2) = c_\btheta\{F_{n1}(x_1),F_{n2}(x_2)\}$, $K_{n,\{1\},\btheta}(x_1,x_2) = \partial_{u_2}C_\btheta\{F_{n1}(x_1),F_{n2}(x_2)\}-\partial_{u_2}C_\btheta\{F_{n1}(x_1-),F_{n2}(x_2)\}$, $K_{n,\{2\},\btheta}(x_1,x_2) = \partial_{u_1}C_\btheta\{F_{n1}(x_1),F_{n2}(x_2)\}-\partial_{u_1}C_\btheta\{F_{n1}(x_1),F_{n2}(x_2-)\}$, and 
$K_{n,\{1,2\},\btheta}(x_1,x_2) =
C_\btheta\{F_{n1}(x_1),F_{n2}(x_2)\}-C_\btheta\{F_{n1}(x_1-),F_{n2}(x_2)\} -
C_\btheta\{F_{n1}(x_1),F_{n2}(x_2-)\}+C_\btheta\{F_{n1}(x_1-),F_{n2}(x_2-)\}$. 
\end{example}

\begin{remark}\label{rem:Li1}
In \cite{Li/Li/Qin/Yan:2020}, instead of using $F_{nj}(X_{ij}-)$, the authors used $F_{nj}(X_{ij}-)+\dfrac{1}{n+1}$. The difference between our pseudo log-likelihood and theirs is negligible. However, our choice seems more natural and simplifies notations in the multivariate case, which was not considered in \cite{Li/Li/Qin/Yan:2020}.
\end{remark}
One can see that computing the pseudo log-likelihoods $\ell_n$ might be cumbersome when $d$ is large. To overcome this problem, we propose to use a pairwise composite pseudo log-likelihood. For a review on this approach in other settings, see, e.g., \cite{Varin/Reid/Firth:2011}. See also \cite{Oh/Patton:2016} for the composite method in a particular copula context.
Here, the pairwise composite pseudo log-likelihood is simply defined as
$\disp
\check \ell_n = \sum_{1\le k < l\le d} \ell_n^{(k,l)}$,
where $\ell_n^{(k,l)}$ is the pseudo log-likelihood defined by \eqref{eq:ll2d} for the pairs $(X_{ik},X_{il})$, $i\in\setn$. One can also replace $\ell_n$ by $\tilde\ell_n$ is the previous expression.
If $\check \btheta_n= \disp \arg\max_{\btheta \in \cP} \check\ell_n(\btheta)$, set $\check \bTheta_n = \sn(\check \btheta_n-\btheta_0)$.
The asymptotic behavior of the estimators $\btheta_n$, $\tilde\btheta_n$, and $\check\btheta_n$ is studied next.

\subsection{Notations and other assumptions}\label{ssec:results}

For sake of simplicity, the gradient column vector of a function $g$ with respect to $\btheta$ is often denoted $\dot g$ or $\nabla_\btheta g$, while the associated Hessian matrix is denoted by $\ddot g$ or $\nabla_\btheta^2 g$.
Before stating the main convergence results, for any $A\subset \setd$, $0\le a_j < b_j \le 1$, $j\in A$, $u_j\in (0,1)$, $j\in A^\complement$, define
\begin{equation}\label{eq:varphi-cop}
\bvarphi_{A,\btheta}\left(\ba,\bb,\bu\right) = \frac{\sum_{B\subset A}(-1)^{|B|}\partial_{A^\complement}\dot C_\btheta\left((\ba,\bb,\bu)^{(B,A)}\right)}{\sum_{B\subset A}(-1)^{|B|}\partial_{A^\complement} C_\btheta\left((\ba,\bb,\bu)^{(B,A)}\right)},
\end{equation}
where for $B\subset A$,
$
\left((\ba,\bb,\bu)^{(B,A)}\right)_j =
\left\{
\begin{array}{ll}
a_j ,& j\in  B, \\
b_j,   & j  \in  A\setminus B,\\
 u_j , & j\in A^\complement.
 \end{array}
 \right.$
In particular, $\bvarphi_\emptyset(\bu) = \frac{\dot c_\btheta(\bu)}{c_\btheta(\bu)}$, and if $A=\{1,\ldots, k\}$, $k\in\setd$, one has
\begin{multline*}
 \bvarphi_{A,\btheta}(\ba,\bb,\bu) = \bvarphi_{A,\btheta}(a_1,\ldots,a_k,b_1,\ldots b_k, u_{k+1},\ldots, u_{d}) \\
 = \frac{\int_{a_1}^{b_1} \cdots \int_{a_k}^{b_k} \dot c_\btheta (s_1,\ldots,s_k,u_{k+1},\ldots,u_d)ds_1 \cdots d s_k}{\int_{a_1}^{b_1} \cdots \int_{a_k}^{b_k}c_\btheta(s_1,\ldots,s_k,u_{k+1},\ldots,u_d)ds_1 \cdots d s_k}.
\end{multline*}
Next, set
$\bH_{A,\btheta}(\bx)=
J_A(\bx) \bvarphi_{A,\btheta}\left\{\bF_{A}(\bx-), \bF_{A}(\bx), \bF_{A^\complement}(\bx) \right\}$.
Finally, set
$\cK=\cK_{\btheta_0}$ and
$\bH_\btheta = \sum_{A\subset \setd} \bH_{A,\btheta}$.
\begin{ass}\label{hyp:varphi}
The functions $\bvarphi_{A,\btheta}$, $A\subset \setd$, satisfy Assumptions  \ref{hyp:varphiApp}--\ref{hyp:est}.
\end{ass}
Next, to be able to deal with $\tilde\ell_n$ and the associated composite pseudo log-likelihood, one needs extra assumptions. First, one needs to control the gradient of the likelihood.
\begin{ass}\label{hyp:max} There is a neighborhood $\cN$ of $\btheta_0$ such that
for any $A\subset \setd$,
$$
\sni\max_{\btheta\in\cN} \max_{1\le i\le n}
\left| \frac{
\sum_{B\subset A} (-1)^{|B|} \dot \partial_{A^\complement}C_{\btheta}\left\{\bF_n^{(B)}(\bX_i)\right\}}{\sum_{B\subset A} (-1)^{|B|} \partial_{A^\complement}C_{\btheta}\left\{\bF_n^{(B)}(\bX_i)\right\}}\right| \stackrel{Pr}{\to} 0.
$$
\end{ass}

Second, in order to measure the errors one makes  by considering $J_{n,A}$ instead of $J_A$, set
$$
\cE_{nj}= \sum_{i=1}^n \I\{\Delta F_{nj}(X_{ij}) = 1/(n+1)\}\I(X_{ij}\in \cA_j).
$$
Then,  $\disp E(\cE_{nj}) = V_n(F_j) = n\sum_{x\in \cA_j}\Delta F_j(x)\{1-\Delta F_j(x)\}^{n-1}$,  $j\in\setd$.
\begin{ass}\label{hyp:support}
For any $j\in\setd$, $\disp
\limsup_{n\to\infty}V_n(F_j) <\infty$.
\end{ass}
Assumption \ref{hyp:support} means that the average  number of indices $i$ so that $ \Delta F_{nj}(X_{ij})=1/(n+1)$ when $X_{ij} \in \cA_j$ is bounded.
This assumption holds true when the discrete part of the margin is a finite discrete distribution, a Geometric distribution, a Negative Binomial distribution, or a Poisson distribution, using Remark 1 
in the Supplementary Material.

\subsection{Convergence of estimators}\label{ssec:conv1}

Recall that  for any $j\in \setd$ and any $i\in \{1,\ldots,n\}$, $X_{ij} = F_{j}^{-1}(U_{ij})$, where
$\bU_1 = (U_{11}, \ldots, U_{1d})$, $\ldots, \bU_n = (U_{n1}, \ldots, U_{nd})$  are iid observations from copula $C_{\btheta_0}$. Also, let $P_{\btheta_0}$ be the associated probability distribution of $\bX_i$, for any $i\in \setn$.\\
Now, for any $j\in\setd$,  and any $y\in \dR$, $
F_{nj}(y) =  B_{nj}\circ F_j(y)$,
where
$\disp
B_{nj}(v) = \dfrac{1}{n+1}\sum_{i=1}^n \I(U_{ij}\le v)$, $v\in [0,1]$.
Note that if for some $i\in \setn$, $j\in\setd$,  $X_{ij}\not\in \cA_j$, i.e., $\Delta F_j(X_{ij})=0$, then $F_{nj}(X_{ij}) = B_{nj}(U_{ij})$.
It is well-known that the processes $\dB_{nj}(u_j) = \sn\{B_{nj}(u_j)-u_j\}$, $u_j\in [0,1]$, $j\in\setd$, converge jointly  in $C([0,1])$ to  $\dB_j$, denoted $\dB_{nj} \rightsquigarrow \dB_j$, where $\dB_j$ are Brownian bridges, i.e., $\dB_1, \ldots, \dB_d$ are continuous centered Gaussian processes with
$
{\rm Cov}\{\dB_j(s),\dB_k(t)\} 
= P_{\btheta_0}(U_{1j}\le s, U_{1k}\le t)-st$, $s,t\in [0,1]$.
In particular, for any $j\in\setd$, ${\rm Cov}\{\dB_j(s),\dB_j(t)\} = \min(s,t)-st$.
Before stating the main convergence results for the estimation errors $\bTheta_n = n^{1/2}(\btheta_n-\btheta_0) $ and
 $\tilde\bTheta_n = n^{1/2}\left(\tilde \btheta_n-\btheta_0\right) $, for any $A\subset \setd$, define $\bzeta_{A,i} = H_{A,\btheta_0}(X_i)$ and set $\disp \bzeta_i = \sum_{A\in\setd}\bzeta_{A,i}$.
Next,  set $\disp \cW_{n,1} = \sni\sum_{i=1}^n \bzeta_i$ and let
\begin{eqnarray*}
\cW_{n,2}
  &=& -  \sum_{j=1}^d \sum_{A\not\ni j} \int c_{\btheta_0}(\bu) \bbeta_{A,j,2}(\bu)\dB_{nj}(u_{j})d\bu\\
&& \quad - \sum_{j=1}^d \sum_{x_j\in \cA_j} \dB_{nj}\{F_j(x_j)\}\sum_{A\ni j}   \int c_{\btheta_0}(\bu)   \bbeta_{A,j,2+}(x_j,\bu)d\bu \\
&& \qquad
 + \sum_{j=1}^d  \sum_{x_j\in \cA_j} \dB_{nj}\{F_j(x_j-)\} \sum_{A\ni j} \int c_{\btheta_0}(\bu)   \bbeta_{A,j,2-}(x_j, \bu)d\bu,
\end{eqnarray*}
where the functions $\bbeta_{A,j}, \bbeta_{A,j,\pm}, \bbeta_{A,j,1,\pm}, \bbeta_{A,j,2\pm}$ are defined in Appendix \ref{app:eta}.
Finally, set $\cW_{n,0} = n^{-1/2}\sum_{i=1}^n \dfrac{\dot c_{\btheta_0}(\bU_i)}{c_{\btheta_0}(\bU_i)}$.
Basically, $\cW_{n,1}$ is what one should have if the margins were known, while $\cW_{n,2}$ is the price to pay for not knowing the margins. Finally, $\cW_{n,0}$ is needed for obtaining bootstrapping results \citep{Genest/Remillard:2008, Nasri/Remillard:2019}.
The following result is  a consequence of  Theorem \ref{thm:gen_est} and  Theorem \ref{thm:zeta}, proven in Appendix \ref{app:Zest} and Appendix \ref{app:zeta} respectively.



\begin{thm}\label{thm:main}
Under Assumptions \ref{hyp:cop}--\ref{hyp:varphi}, $\bTheta_n $ converges in law to $\bTheta = \cJ_1^{-1}(\cW_1+\cW_2)$,
where $\cW_1 \sim N(0,\cJ_1)$, $E\left(\cW_1 \cW_0^\top\right) = \cJ_1$, and  $E\left(\cW_2 \cW_0^\top\right) = 0$.
If in addition Assumptions \ref{hyp:max}--\ref{hyp:support} hold true,
 then $\tilde\bTheta_n-\bTheta_n$ converges in probability to $0$, so $\tilde\bTheta_n$ converges in law to $\bTheta$.
\end{thm}

\begin{corollary}\label{cor:reg}
$\btheta_n$ and $\tilde \btheta_n$ are regular estimator of $\btheta_0$ in the sense that $\bTheta_n$ and $\tilde \bTheta_n$ converge to $\bTheta$ and $E\left(\bTheta \cW_0^\top \right) = I_d$.
\end{corollary}
\begin{proof}
Theorem  \ref{thm:main} yields
$ E\left( \bTheta \cW_0^\top \right)  = \cJ_1^{-1}E \left\{ (\cW_1+\cW_2) \cW_0^\top \right\} = I_d$.
\end{proof}
 \begin{remark}
 As shown in \cite{Genest/Remillard:2008}, regularity of estimators is necessary for parametric bootstrap to work,  using LeCam's third lemma.
 \end{remark}

Having got the asymptotic behavior of $\btheta_n$ and  $\tilde\btheta_n$, it is easy to obtain the following result.

\begin{corollary}\label{cor:LLcomp}
Under Assumptions \ref{hyp:cop}--\ref{hyp:varphi},
$\check \bTheta_n = n^{1/2}(\check\btheta_n-\btheta_0) $ converges in law to
$$
\check\bTheta = \left(\sum_{1\le k<l\le d} \cJ^{(k,l)}\right)^{-1} \sum_{1\le k<l\le d} \cW^{(k,l)},
$$
where $\cW^{(k,l)} = \cW_1^{(k,l)}+ \cW_2^{(k,l)}$ are defined as $\cW_1$ and $\cW_2$ but restricted to the pairs $(X_{ik},X_{il})$, $i\in\setn$. Moreover, $\check\btheta_n$ is regular. The same result holds if $\ell_n$ is replaced by $\tilde\ell_n$, and if in addition,  Assumptions \ref{hyp:max}--\ref{hyp:support} hold true.
\end{corollary}

\begin{remark}\label{rem:IFM}
The previous results hold if one uses parametric margins instead of nonparametric margins, provided the margins are smooth enough and the estimated parameters converge in law. In fact, assume that for $j\in\setd$, $F_{nj} = F_{j,\bgamma_{nj}}$, $F_{j} = F_{j,\bgamma_{0j}}$,  and
$\dF_{nj}=\sn \left\{F_{nj} - F_{j}\right\}$ converges in law to $\dF_j = \bGamma_j^\top \dot F_j$, with
$\dot F_j = \left. \nabla_{\bgamma_j}F_{j,\bgamma_j}\right|_{\bgamma_j=\bgamma_{0j}}$, and $\bGamma_j$ is a centered Gaussian random vector.
Then, under Assumptions \ref{hyp:cop}--\ref{hyp:varphi}, replacing respectively $\dB_{j}(u_j)$, $\dB_{j}\{F_j(x_j)\}$, and $\dB_{j}\{F_j(x_j-)\}$ with
$\dF_{j}\left\{F_j^{-1}(u_j)\right\}$,
$\dF_{j}(x_j)$, and $\dF_{j}(x_j-)$, one obtains the analogs of Theorem \ref{thm:main} and Corollaries \ref{cor:reg}--\ref{cor:LLcomp}. Furthermore, this shows that one could also consider pair copula models. 
\end{remark}

\section{Numerical experiments}\label{sec:exp}

In this section, we study the quality  and performance of the proposed estimators,  
for various choices of copula families, margins and sample sizes.
Note that all simulations were done with the more demanding case $\tilde\btheta_n$. 
In the first set of experiments, we deal with the bivariate case,  and in a second set of experiments, we compute the composite estimator in a trivariate setting.
First, in the bivariate case, we consider five copula families and  five pairs of margins for each copula family.  For the first experiment (Exp1), both margins are standard Gaussian, i.e.,
 $F_1,F_2\sim N(0,1)$. In the second experiment (Exp2), the margins are Poisson with parameters $5$ and $10$ respectively, i.e., $F_1\sim \cP(5)$ and $F_2\sim \cP(10)$, while in the third experiment (Exp3), $F_1 \sim \cP(10)$ and $F_2\sim N(0,1)$.
In the fourth experiment (Exp4), $F_1$ is a rounded Gaussian, namely $X_1 = \lfloor 1000 Z_1\rfloor$,  with $Z_1\sim N(0,1)$, and $F_2\sim N(0,1)$.
Finally, for the fifth experiment (Exp5), $F_1$ is zero-inflated, with $F_1(0-)=0$, $F_1(x) = 0.05+0.95(2F_2(x)-1)$, $x\ge 0$, and $F_2\sim N(0,1)$.
To estimate the parameter $\theta$ of the Clayton, Frank, Gumbel, Gaussian and Student (with $\nu=5$) copula families corresponding to a Kendall's tau $\tau_0 = 0.5$, samples of size $n\in\{100,250,500\}$ were generated for each copula family, and $\btheta_n$ was computed. Here, $\tau$ is Kendall's tau of the copula family $C_\theta$, written $\tau(C_\theta)$. For results to be comparable throughout copula families, we computed the relative
  bias and the relative root mean square error (RMSE) of $\tau(C_{\theta_n})$, instead of $\theta_n$. The results for $1000$  samples are reported in Table \ref{tab:est2d}.
As one can see, the estimator performs quite well for the five numerical experiments and the five copula families. Furthermore, the precision depends on the copula  family, but for a given copula family, the precision does not vary significantly with the margins.
Even the case when both margins are continuous (Exp1) does not yield the best results. When the sample size is 250 or more, the relative bias is always smaller than 2\%. Finally, as expected, both the bias and the RMSE decrease when the sample size increases.
\begin{table}
\caption{\label{tab:est2d} Relative bias and relative RMSE (in parentheses) in percent for  $\tau(C_{\theta_n})$ vs $\tau_0$ when $n\in \{100,250,500\}$, based  on 1000 samples. For the Student copula,  $\nu=5$ is assumed known.}
\centering
{
\begin{tabular}{crrrrr}
\hline
& \multicolumn{5}{c}{Copula}\\
  \cline{2-6}
Margins     & \multicolumn{1}{c}{Clayton}  & \multicolumn{1}{c}{Frank}  & \multicolumn{1}{c}{Gumbel} & \multicolumn{1}{c}{Gaussian}  & \multicolumn{1}{c}{Student} \\[6pt]
 & \multicolumn{5}{c}{$n=100$} \\
\hline
 Exp1 &  -0.42  (9.62) &  0.40   (9.30) &   3.02    (10.9) &  2.92    (9.40) &  2.18   (11.1)           \\
 Exp2 &   0.52  (10.2) &  0.86   (9.64) &   3.70    (11.4) &  3.16    (9.80) &  2.64   (11.5)          \\
 Exp3 &   0.02  (9.84) &  0.58   (9.40) &   3.34    (11.1) &  2.96    (9.52) &  2.36   (11.3)          \\
 Exp4 &  -0.44  (9.62) &  0.40   (9.30) &   3.02    (10.9) &  2.88    (9.38) &  2.14   (11.1)          \\
 Exp5 &  -1.00  (9.90) &  0.36   (9.30) &   2.76    (10.8) &  2.42    (9.28) &  1.68   (11.0)          \\[6pt]
 & \multicolumn{5}{c}{$n=250$} \\
\hline
 Exp1 &   0.18    (6.06)&  -0.32    (5.76) &  0.56    (6.12) & 1.30  (5.92) &  0.92    (6.52)  \\
 Exp2 &   0.50    (6.44)&  -0.00    (5.92) &  1.52    (6.36) & 1.72  (6.18) &  1.52    (6.82) \\
 Exp3 &   0.40    (6.22)&  -0.20    (5.84) &  0.96    (6.16) & 1.42  (6.02) &  1.16    (6.60)  \\
 Exp4 &   0.12    (6.08)&  -0.32    (5.76) &  0.60    (6.12) & 1.30  (5.92) &  0.92    (6.52)  \\
 Exp5 &  -0.10    (6.22)&  -0.34    (5.76) &  0.42    (6.12) & 1.06  (5.90) &  0.62    (6.56)  \\[6pt]
 & \multicolumn{5}{c}{$n=500$} \\
\hline
 Exp1 & 0.08    (4.44) &  -0.04    (3.88) &   0.38      (4.44)&  0.64     (4.40)  &   0.34    (4.48) \\
 Exp2 & 0.36    (4.60) &   0.06    (4.02) &   0.72      (4.64)&  0.76     (4.30)  &   0.54    (4.68) \\
 Exp3 & 0.28    (4.50) &  -0.00    (3.90) &   0.52      (4.52)&  0.70     (4.24)  &   0.42    (4.56) \\
 Exp4 & 0.04    (4.44) &  -0.04    (3.88) &   0.40      (4.44)&  0.62     (4.19)  &   0.32    (4.48) \\
 Exp5 & 0.10    (4.48) &  -0.04    (3.88) &   0.34      (4.44)&  0.54     (4.20)  &   0.24    (4.48) \\
  \hline
\end{tabular}
}
\end{table}
In the second set of experiments, using the pairwise composite estimator $\check\btheta_n$, we estimated the parameters of a trivariate non-central squared Clayton copula \citep{Nasri:2020} with Kendall's tau $\tau_0=0.5$, and non-centrality parameters $a_{10} = 0.9$, $a_{20} = 2.3$, and $a_{30}= 1.4$. In this case, for all five experiments, $F_1$ and $F_2$ are defined as before, while $F_3\sim N(0,1)$. The results are displayed in Table \ref{tab:est3d}. As one could have guessed, the estimation of $\tau$ is not as good as in the bivariate case,  but the results are good enough. As for the non-centrality parameters, the estimation of $a_2$, which has a large value (the upper bound being $3$), is not as good as the other values, but this is coherent with the simulations in \cite{Nasri:2020}. All in all, the composite method yields quite satisfactory results.

\begin{table}[ht!]
\caption{Relative bias and RMSE (in parentheses) in percent for the estimation errors of $\tau_0$  and $a_j$ for $j=1,2,3$ in the case of the trivariate non-central squared Clayton copula when $n\in \{100, 250, 500\}$, based  on 1000 samples.}\label{tab:est3d}
\begin{tabular}{crrrr}
\hline
& \multicolumn{4}{c}{Parameters}\\
\cline{2-5}
Margins   & \multicolumn{1}{c}{$\tau $}  & \multicolumn{1}{c}{$a_1$}  & \multicolumn{1}{c}{$a_2$}  & \multicolumn{1}{c}{$a_3$}   \\[6pt]
 & \multicolumn{4}{c}{$n=100$} \\
\hline
 Exp1 &     4.01  (11.21) &   -4.92   (24.95) &    -30.67   (41.09) &    -5.08    (21.56)    \\
 Exp2 &     4.24  (11.68) &   -4.70   (25.27) &    -28.71   (40.58) &    -4.83    (21.33)    \\
 Exp3 &     4.31  (11.50) &   -5.74   (23.72) &    -30.02   (40.95) &    -5.32    (21.08)     \\
 Exp4 &     3.99  (11.22) &   -4.93   (25.04) &    -30.74   (41.10) &    -5.12    (21.52)     \\
 Exp5 &     3.99  (11.18) &   -4.98   (24.34) &    -30.63   (41.12) &    -5.22    (21.18)     \\[3pt]
 & \multicolumn{4}{c}{$n=250$} \\
\hline
 Exp1 &  1.34    (6.80)  & -4.10 (13.05) &  -21.93   (34.13)   &  -3.82     (13.07)    \\
 Exp2 &  1.32    (6.97)  & -2.82 (13.24) &  -20.59   (33.63)   &  -3.15     (13.07)    \\
 Exp3 &  1.41    (6.85)  & -3.47 (13.47) &  -20.46   (33.45)   &  -3.66     (13.08)     \\
 Exp4 &  1.33    (6.77)  & -4.01 (12.98) &  -22.01   (34.03)   &  -3.77     (12.83)     \\
 Exp5 &  1.34    (6.80)  & -4.10 (13.04) &  -21.93   (34.14)   &  -3.82     (13.08)     \\[3pt]
  & \multicolumn{4}{c}{$n=500$} \\
\hline
 Exp1 &    0.63    (4.56) &   -2.07    (8.81) &   -12.30   (25.44) &  -2.19    (8.09)    \\
 Exp2 &    0.74    (4.74) &   -1.98    (9.13) &   -12.81   (26.67) &  -2.19    (8.54)    \\
 Exp3 &    0.73    (4.61) &   -2.07    (8.93) &   -12.24   (25.45) &  -2.30    (8.20)     \\
 Exp4 &    0.63    (4.56) &   -2.06    (8.84) &   -12.31   (25.44) &  -2.21    (8.13)     \\
 Exp5 &    0.63    (4.56) &   -2.05    (8.78) &   -12.20   (25.33) &  -2.17    (8.06)     \\
  \hline
\end{tabular}
\end{table}

\section{Example of application}\label{sec:example}

In this section, we propose a rigorous method to study the relationship between duration and severity for hydrological data used in \cite{Shiau:2006}. The data were kindly provided by the author.
There are many articles in the hydrology literature about modeling drought duration and severity with copulas; see, e.g., \cite{Chen/Singh/Guo/Mishra/Guo:2013, Shiau:2006}. One of the main tools to compute the drought duration and severity is the so-called Standardized Precipitation Index (SPI) \citep{McKee/Doesken/Kleist:1993}. Basically, \cite{McKee/Doesken/Kleist:1993} suggest to fit a gamma distribution over a moving average (1-month, 3-month, etc.) of the precipitations and then transform them into a Gaussian distribution. However, it may happen that there are several zero values in the observations so fitting a continuous distribution is not possible.
Using the data kindly provided by Professor Shiau (daily precipitations in millimeters for the Wushantou gauge station from 1932 to 2001), we see from Figure \ref{fig:densityMAP} that even taking a 1-month moving average leads to a zero-inflated distribution. So, instead of fitting a gamma distribution to the moving average, as it is often done, we suggest to simply apply the inverse Gaussian distribution to the empirical distribution. Then, one can compute the duration and severity: a drought is defined as a  sequence of consecutive days with negative SPI values, say $SPI_{i}, \dots, SPI_j$: the length the sequence is the duration $D$, i.e., $D=j-i+1$,  and the severity is defined by $S = -\sum_{k=i}^j SPI_k$.  It makes sense to consider the severity $S$ as a continuous random variable but the duration $D$ is integer-valued. Again, in the literature, a continuous distribution is usually fitted to $D$, which is incorrect. These variables are then divided by $30$ in order to represent months. With the dataset, we obtained 175 drought periods.
A non-parametric estimation of the density of the severity per month is displayed in Figure \ref{fig:densityMAP} which seems to be a mixture of at least two distributions. We tried mixtures of up to $4$ gamma distributions without success. A scatter plot of the duration and the severity also appears in Figure \ref{fig:densityMAP}. With the copula-based methodology developed here, based on a measure of fit, we chose the Frank copula.
In contrast, the preferred copula families in \cite{Shiau:2006} were the Galambos and Gumbel families. Using a smoothed distribution for the severity $D$, we can compute the conditional probability  $P(D > y|S=s)$ for $y=1$ to $8$ months, in addition to the conditional expectation $E(D|S=s)$. These functions are displayed in Figure \ref{fig:condDuration}.

\begin{figure}[ht!]
    \centering
    \includegraphics[scale=0.25]{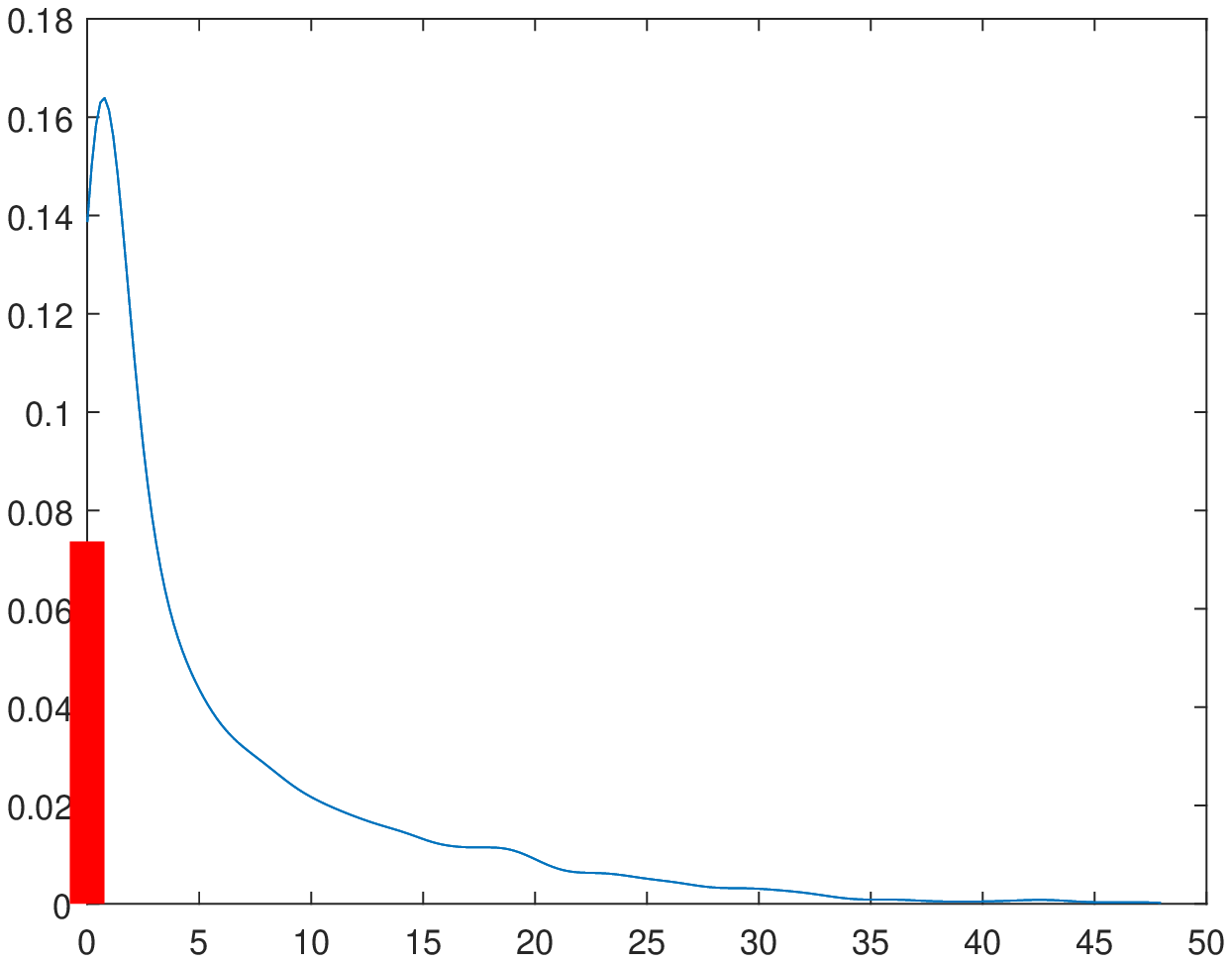}
    \includegraphics[scale=0.25]{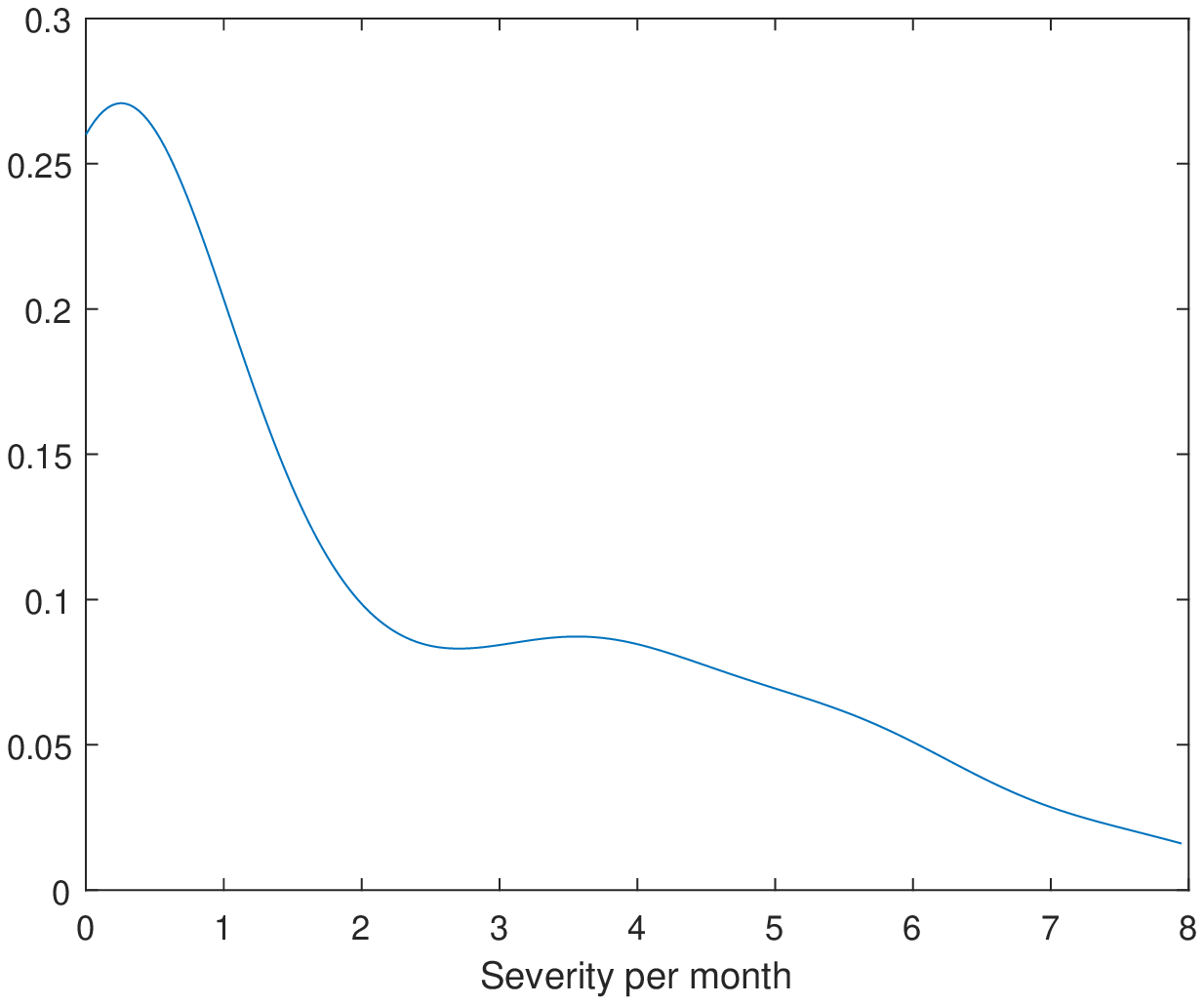}
      \includegraphics[scale=0.25]{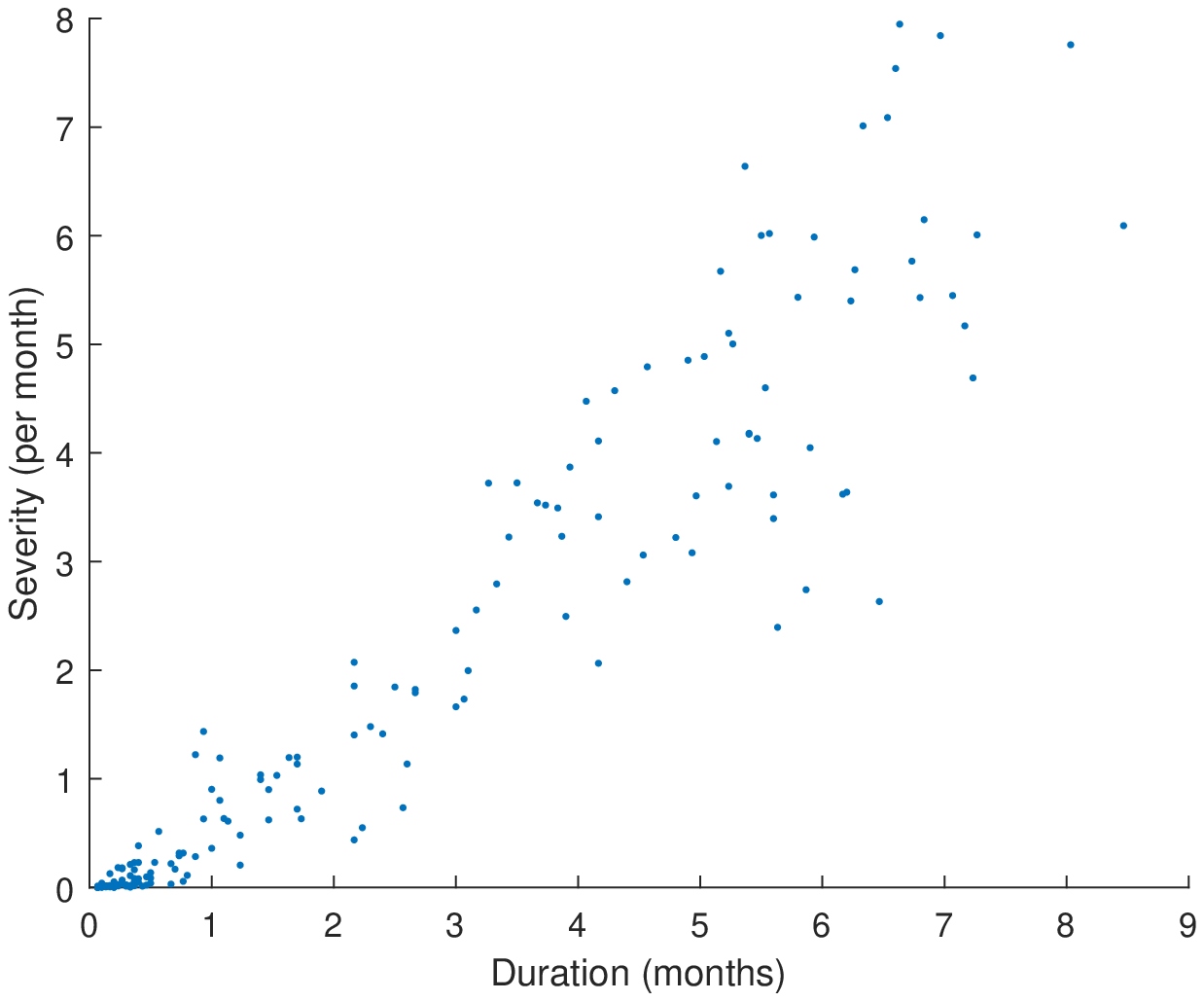}
    \caption{Estimated zero-inflated density for the 1-month moving average of precipitations (top left) and severity per month (top right), together with a scatter plot of the duration and severity (bottom).}
    \label{fig:densityMAP}
\end{figure}
\begin{figure}[ht!]
    \centering
   \includegraphics[scale=0.25]{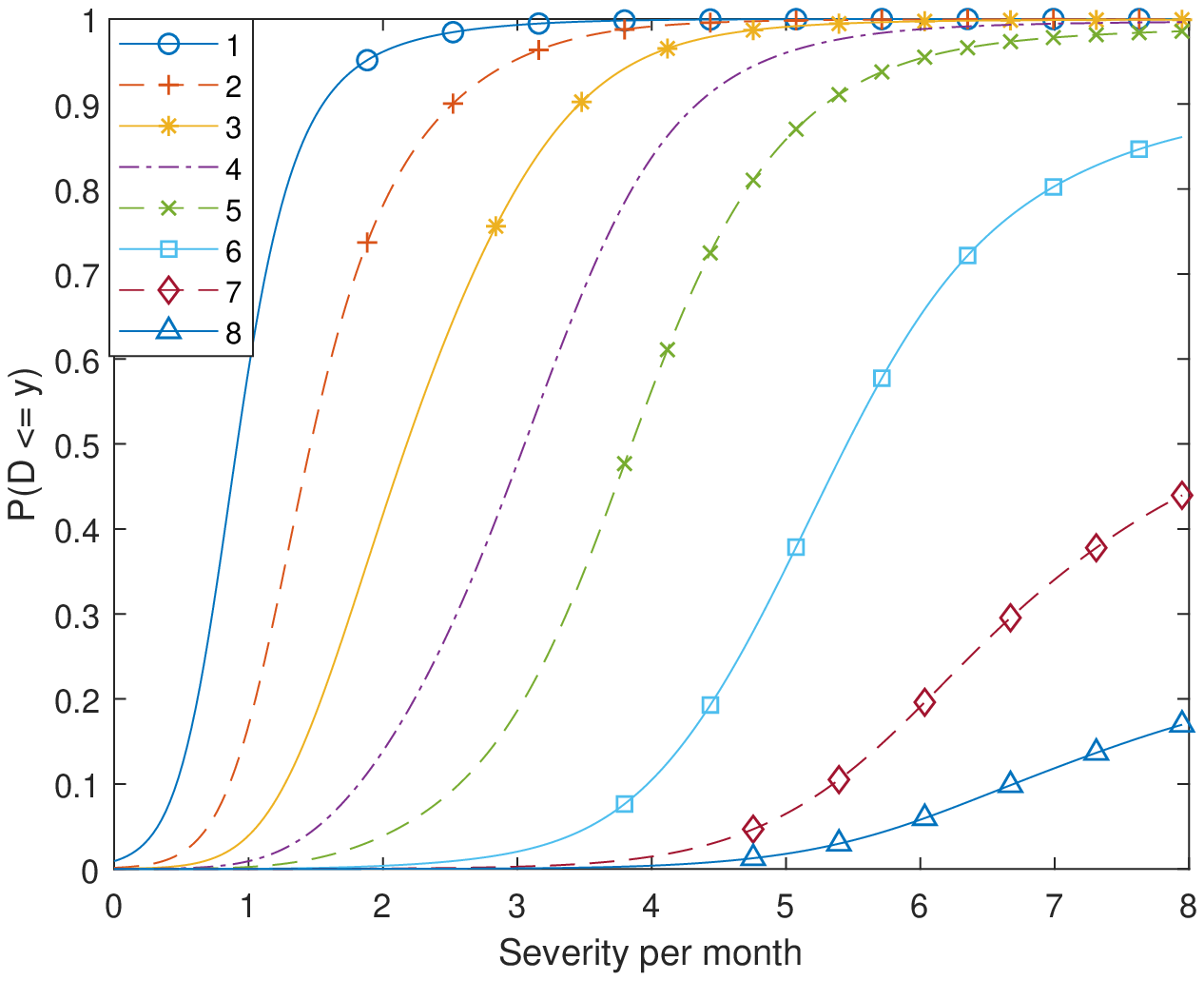}
\includegraphics[scale=0.25]{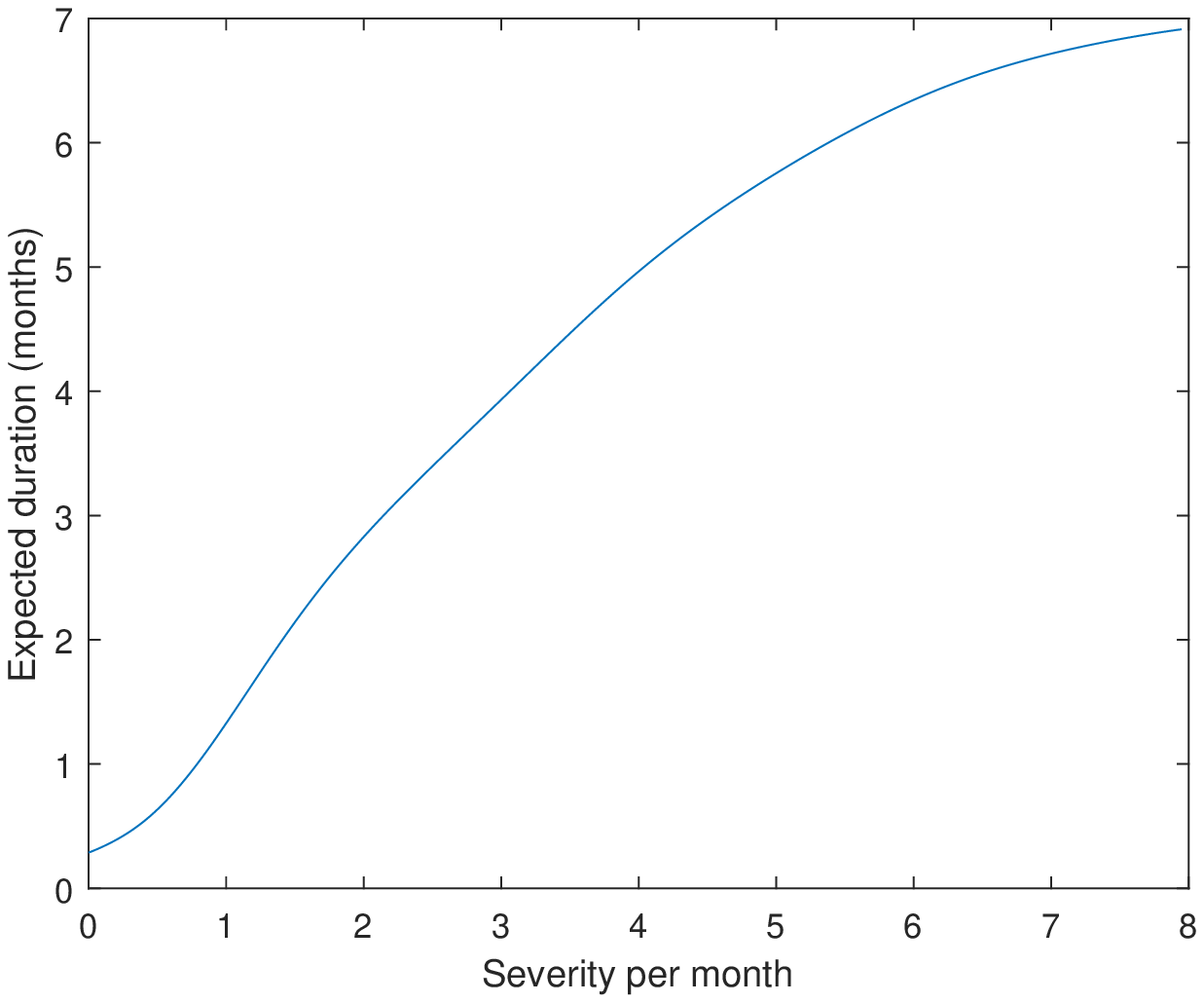}
    \caption{Conditional probability $P(D \le y)$ for $y\in \{1,\ldots,8\}$ months and conditional expectation of the duration given severity per month.}
    \label{fig:condDuration}
\end{figure}

\section{Conclusion}

We presented methods based on pseudo log-likelihood for estimating the parameter of copula-based models for arbitrary multivariate data.
These pseudo log-likelihoods depend on the non-parametric margins and are adapted to take into account the ties. We have also shown that the methodology can be extended to the case of parametric margins. According to numerical experiments, the proposed estimators perform quite well. As a example of application, we
estimated the relationship between drought duration and severity in hydrological data, where the problem of ties if often ignored.  The proposed methodologies can also be applied to high dimensional data. For this reason, we have shown in Corollary \ref{cor:LLcomp} that
the pairwise composite applied to our bivariate pseudo-likelihood is valid.
Finally, in a future work, we will also develop bootstrapping methods and formal tests of goodness-of-fit.

\section*{Acknowledgements}
The first author is supported by the Fonds de re\-cher\-che du Qu\'ebec -- Natur\-e et
tech\-no\-lo\-gies,  the Fonds de re\-cher\-che du Qu\'ebec -- santé, the \'Ecole de santé publique de l'Universit\'e de Montr\'eal, the Centre de recherche en sant\'e publique (CRESP), and the Natural Sciences and Engineering Research Council of Canada.
The second author is supported by the Natural Sciences and Engineering Research of Canada. 
 We would also like to thank Professor Jenq-Tzong Shiau of National Cheng Kung University for sharing his data with us.

\begin{appendix}

\section{Rank-based Z-estimators when margins are arbitrary}\label{app:Zest}
Suppose that $\bX_i=\bF^{-1}(\bU_i)$, where $\bU_1,\ldots, \bU_n$ are iid observations from copula $C=C_{\btheta_0}$.
 For any $A\subset \setd$,  set $J_A(\bx)=  \left[\prod_{j\in A} \I\{x_j\in \cA_j \}\right] \left[\prod_{j\in A^\complement} \I\{x_j\not\in \cA_j \}\right]$. Also, the set of atoms $\cA_j$ of $F_j$ is assumed to be closed for any $j\in\setd$. This implies that if $x_n\downarrow x $, and $\Delta F_j(x_n)>0$, then $\Delta F_j(x)>0$.
 Further set
$\disp \bH_{n,A,\btheta}(\bx) =
J_A(\bx) \bvarphi_{A,\btheta}\left\{\bF_{n,A}(\bx-), \bF_{n,A}(\bx), \bF_{n,A^\complement}(\bx) \right\}= J_A(\bx) \varphi_{A,\btheta}\left\{\bF_{n,A}(\bx-), \bF_{n}(\bx) \right\}$, 
and
$\disp
\bH_{A,\btheta}(\bx)=
J_A(\bx) \bvarphi_{A,\btheta}\left\{\bF_{A}(\bx-), \bF_{A}(\bx), \bF_{A^\complement}(\bx) \right\}=\\ J_A(\bx)  \bvarphi_{A,\btheta}\left\{\bF_{A}(\bx-), \bF(\bx) \right\}$,
where $\bF_{n,A}(\bx-)= \{F_{nj}(x_j-): j\in A\}$ and $\bF_A(\bx-)= \{F_j(x_j-): j\in A\}$. Here, $\bvarphi_{A,\btheta}$ is a $\dR^p$-value continuously differentiable function defined on $\{(\ba,\bb,\bu): 0\le a_j < b_j\le 1, \; j\in A, \; 0<u_j<1,\; j\not\in A\}$. To simplify notations, one also writes $\varphi_{A,\btheta}\left(\ba,\bu\right)$ if $b_j=u_j$ for all $j\in A$.
Finally, set
$\cK=\cK_{\btheta_0}$, $\disp
\bH_{n,\btheta}(\bx) = \sum_{A\subset \setd} \bH_{n,A,\btheta}(\bx)$, and $\disp \bH_\btheta(\bx) = \sum_{A\subset \setd} \bH_{A,\btheta}(\bx)$.
\begin{example}\label{ex:spearman}
Take $\varphi_A(\ba,\bb,\bu) = \left\{\prod_{j\in A}\left( \frac{a_j+b_j}{2} - \dfrac{1}{2}\right)\right\}
\left\{\prod_{j\in A^\complement}\left(u_j - \dfrac{1}{2}\right)\right\}$. This can be used to define Spearman's rho.
\end{example}
\begin{example}[Estimation of copula parameter]\label{ex:est-cop}
In this case, $\varphi_{A,\btheta}$ is defined by \eqref{eq:varphi-cop}.
\end{example}
Next, define $\bT_n(\btheta) = \int\bH_{n,\btheta} d\cK_n$ and $\bmu(\btheta) = \int  \bH_\btheta d\cK $. We are interested in the convergence in law of $\dT_n (\btheta) = \sn\{\bT_n(\btheta)-\bmu(\btheta)\}$, which is a multivariate extension of \cite{Ruymgaart/Shorack/vanZwet:1972} to arbitrary distributions,
where \cite{Ruymgaart/Shorack/vanZwet:1972} considered the bivariate case with continuous margins, and where $\varphi$ is a product and independent of $\btheta$. Before stating the result, one  needs to define the following class of functions.

\begin{definition}\label{def:func}
Let $\cQ_p$ is the set of all positive continuous functions $q$ on $(0,1)^p$, such that for any $\omega\in (0,1)$, the exists a constant $c(\omega)$ independent of $\bu\in (0,1)^p$ so that
\begin{equation}\label{eq:Qomega}
\sup_{\substack{\omega u_j \le t_j \\ \omega(1-u_j) \le 1-t_j\\ j\in \setp}}  q(\bt)\le c(\omega)q(\bu).
\end{equation}
Note that \eqref{eq:Qomega} is an extension of reproducing u-shaped functions defined in \cite{Tsukahara:2005}.   One can see that
$\cQ_p$ is closed under finite sums and finite products. In particular, if $q_j\in \cQ_1$ for all $j\in \setp$, then
$\disp  q(\bt) = \sum_{j=1}^p q_j(t_j) \in \cQ_p$ and
$ \disp  q(\bt) =  \prod_{j=1}^p q_j(t_j) \in \cQ_p$.
\end{definition}
\begin{remark}
Note that $r(t) = r_0(t) r_1(1-t) \in \cQ_1 $, if $r_0, r_1$ are positive, continuous, non-increasing functions such that  $r_i(\omega t)\le c_i(\omega)r_i(t)$, $\omega,t\in (0,1)$, $i\in\{0,1\}$.
For example, $q(t) = 1-\log{t}$ and $q(t)=t^{-a}(1-t)^{-b}$ belong to $\cQ_1$ if $a,b\ge 0$. In \cite{Ruymgaart/Shorack/vanZwet:1972}, it was assumed that $a=b$.

\end{remark}
For any $\gamma\in (0,0.5)$, set
$\disp \cU_\gamma= \{(0,b): \gamma < b < 1-\gamma\}\cup \{(a,b): \gamma < a < b <1-\gamma\}\cup\{(a,1): \gamma < a < 1-\gamma\}$,
and define $r_\epsilon(t) = t^\epsilon(1-t)^\epsilon$, $t\in [0,1]$. The following hypotheses are  needed in the proof.

\begin{ass}\label{hyp:varphiApp}
For any $A\subset\setd$, $\gamma\in(0,1/2)$, there exists a neighborhood $\cN$ of $\btheta_0$ and a constant $\kappa_{\gamma,\cN}$ such that for any $j\in A$, $l\in A^\complement$, and $\bu \in (0,1)^{A^\complement}$,
\begin{eqnarray}
\sup_{\btheta\in \cN} \sup_{(\ba,\bb)\in \cU_\gamma^{|A|}}|\bvarphi_{A,\btheta}(\ba,\bb,\bu)| &\le &\kappa_{\gamma,\cN} \; q_{A^\complement}(\bu),\label{eq:H1}\\
\sup_{\btheta\in \cN} \sup_{(\ba,\bb)\in \cU_\gamma^{|A|}}| \partial_{a_j}\bvarphi_{A,\btheta}(\ba,\bb,\bu)| &\le &\kappa_{\gamma,\cN} \; q_{A^\complement}(\bu),\label{eq:H2}\\
\sup_{\btheta\in \cN} \sup_{(\ba,\bb)\in \cU_\gamma^{|A|}}| \partial_{b_j}\bvarphi_{A,\btheta}(\ba,\bb,\bu)| &\le &\kappa_{\gamma,\cN}\; q_{A^\complement}(\bu),\label{eq:H3}\\
\sup_{\btheta\in \cN} \sup_{(\ba,\bb)\in \cU_\gamma^{|A|}}|\partial_{u_l}\bvarphi_{A,\btheta}(\ba,\bb,\bu)| &\le &\kappa_{\gamma,\cN} \; q_{l,A^\complement}(\bu) , \label{eq:H4}
\end{eqnarray}
where $q_{A^\complement}\in \cQ_{A^\complement}$ is integrable with respect to $C=C_{\btheta_0}$, and  for some $0<\epsilon<0.5$, $r_\epsilon(u_l)  q_{l,A^\complement}$ is integrable with respect to $C$ and $r_\epsilon(u_l)  q_{l,A^\complement}\in \cQ_{A^\complement}$. Also, $\sup_{\btheta\in \cN}|\bH_\btheta|$ is square integrable with respect to $K$.
\end{ass}
\begin{remark} If $q_A(\bu) = \sum_{j\in A}r_j(u_j)$, with $r_j$ integrable for any $j\in \setd$, then $q_A$ is integrable with respect to any copula, since its margins are uniform!
In the copula literature, when the margins are continuous, as well as in \cite{Ruymgaart/Shorack/vanZwet:1972} for $\cN=\{\btheta_0\}$, all results about the convergence are proven on the assumption that $q(\bu) = \prod_{j=1}^d r_j(u_j)$, $r_j$
symmetric, \citep{Genest/Ghoudi/Rivest:1995,Fermanian/Radulovic/Wegkamp:2004,Tsukahara:2005}.
In some cases, those hypotheses are too restrictive, for example for Clayton, Gaussian, and Gumbel families.
\end{remark}
Before stating the main result, define $\dT_{1,n}(\btheta)= \sn \int  \bH_\btheta(d\cK_n-d\cK) $, and set  $\disp  \dT_{2,n}(\btheta) = \sum_{A\in\setd} \int \dH_{n,A,\btheta} d\cK$,
where
\begin{eqnarray*}
\dH_{n,A,\btheta}(\bx) &=& J_A(\bx) \sum_{j\in A} \dF_{nj}(x_j)\partial_{b_j}\bvarphi_{A,\btheta}\{\bF_A(\bx-),\bF(\bx)\}  \I\{F_{nj}(x_{j})<1\} \\
&& \qquad
+ J_A(\bx) \sum_{j\in A}  \dF_{nj}(x_j-)\partial_{a_j} \bvarphi_{A,\btheta}\{\bF_A(\bx-),\bF(\bx)\}\I\{F_{nj}(x_{j}-)>0\} \\
&& \qquad \qquad  +J_A(\bx) \sum_{j\in A^\complement} \dF_{nj}(x_j)\partial_{u_j}\bvarphi_{A,\btheta}\{\bF_A(\bx-),\bF(\bx)\}.
\end{eqnarray*}

\begin{thm}\label{thm:rum_ext}
Under Assumption \ref{hyp:varphiApp}, 
 there exists a neighborhood $\cN$ of $\btheta_0$ such that as $n\to\infty$,
$$
\sup_{\btheta\in\cN}\left|\dT_{n}(\btheta) -\dT_{1,n}(\btheta)-\dT_{2,n}(\btheta)\right| \stackrel{Pr}{\to} 0.
$$
Furthermore, $\dT_{1,n}$, $\dT_{2,n}$, and $\dT_n$ converge jointly in $C(\cN)$ to continuous centered Gaussian processes  $\dT_{1}$, $\dT_{2}$, and $\dT$, where $\dT = \dT_1+\dT_2$, and
\begin{multline*}
 \dT_{2}(\btheta) =  \sum_{A\subset \setd} \sum_{j\in A}\int J_A(\bx)\dF_{j}(x_j)\partial_{b_j} \bvarphi_{A,\btheta}\{\bF_A(\bx-),\bF(\bx)\} d\cK(\bx)\\
 +  \sum_{A\subset \setd} \sum_{j\in A}\int J_A(\bx) \dF_{j}(x_j-)\partial_{a_j} \bvarphi_{A,\btheta}\{\bF_A(\bx-),\bF(\bx)\} d\cK(\bx)\\
 +\sum_{A\subset \setd} \sum_{j\in A^\complement}\int J_A(\bx)\dF_{j}(x_j)\partial_{u_j} \bvarphi_{A,\btheta}\{\bF_A(\bx-),\bF(\bx)\} d\cK(\bx).
\end{multline*}
\end{thm}
\begin{proof}
Set
$S_\gamma = \left\{\bx = (x_1,\ldots,x_d)\in\dR^d : \gamma \le F_j(x_j), F_j(x_j-)\le 1-\gamma\right\}$.
On $S_\gamma$ and $\btheta\in\cN$, for some neighborhood $\cN$ of $\btheta_0$,
$\sn(\bH_{n,\btheta}-\bH_\btheta)-\dH_{n,\btheta}$
converges uniformly in probability to $0$. As a result,
\begin{eqnarray*}
\dT_n(\btheta) &=& \sn \{\bT_n(\btheta)-\bmu(\btheta)\} = \dT_{1,n}(\btheta)+ \sn \int (\bH_{n,\btheta}-\bH_\btheta)d\cK_n \\
&= & \dT_{1,n}(\btheta)+\dT_{2,n}(\btheta)+ \int_{S_\gamma} \left\{\sn (\bH_{n,\btheta}-\bH_\btheta)-\dH_{n,\btheta}\right\}d\cK_n
+ \sn \int_{S_\gamma^\complement} (\bH_{n,\btheta}-\bH_\btheta) d\cK_n\\
&& \qquad + \int_{S_\gamma} \dH_{n,\btheta} (d\cK_n-d\cK) - \int_{S_\gamma^\complement} \dH_{n,\btheta}d\cK\\
&= & \dT_{1,n}(\btheta)+\dT_{2,n}(\btheta) + \sn \int_{S_\gamma^\complement}  (\bH_{n,\btheta}-\bH_\btheta)d\cK_n -\int_{S_\gamma^\complement} \dH_{n,\btheta} d\cK+o_P(1).
\end{eqnarray*}
For  $\epsilon \in \left(0,1/2\right)$ fixed, set
$\cB_{n,M,\epsilon} = \bigcap_{j=1}^d \left\{ \sup_{0<u<1} \dfrac{|\beta_{n,j}(u)|}{r_0^\epsilon(u)} \le M\right\}\cap \left\{ \sup_{0<u<1} |\beta_{n,j}(u)| \le M\right\}$,
where $r_0(u)=u(1-u)$, $u\in [0,1]$. Then, for $\delta>0$ given, one can find $M$ so that for any $n\ge 1$, $P(\cB_{n,M,\epsilon})>1-\delta/2$.  Then, setting $D_{1,n} = \sup_{\btheta\in\cN} \disp \int_{S_\gamma^\complement}\left| \dH_{n,\btheta}\right| d\cK$, one gets
\begin{multline*}
  P\left(\left|D_{1,n}\right|> \delta\right) \le \delta/2  +
P\left(\cB_{n,M,\epsilon} \cap \left\{\left| D_{1,n}\right|> \delta\right\}\right)\\
\le \delta/2 +
\sum_{A\in \setd}\left\{ \sum_{j\in A} (|B_{A,j1}|+|B_{A,j2}|)+\sum_{j\in A^\complement} |B_{A,j3}|\right\},
\end{multline*}
where, for any $j\in\setd$,
$$
B_{A,j1} = M \int_{S_\gamma^\complement}J_A(\bx)\sup_{\btheta\in\cN} \left|\partial_{b_j}  \bvarphi_{A,\btheta}\{\bF_A(\bx-),\bF(\bx)\}\right| \I\{F_{j}(x_j)<1\} d\cK(\bx),
$$
$$
B_{A,j2} =  M \int_{S_\gamma^\complement}J_A(\bx) \sup_{\btheta\in\cN} \left|\partial_{b_j}  \bvarphi_{A,\btheta}\{\bF_A(\bx-),\bF(\bx)\}\right| \I\{F_{j}(x_{j}-)>0\} d\cK(\bx),
$$
$$
B_{A,j3} = \int_{S_\gamma^\complement} J_A(\bx)  r_0^\epsilon \{F_j(x_j)\}\sup_{\btheta\in\cN}\left|\partial_{u_j}\bvarphi_{A,\btheta}\{\bF_A(\bx-),\bF(\bx)\}\right| d\cK(\bx),
$$
since $\dF_{nj}(x_j\pm) = \beta_{nj}\circ F_j(x_j\pm)$ a.s., $j\in \setd$.
By hypothesis \eqref{eq:H2}--\eqref{eq:H3}, it follows that for $i\in\{1,2\}$, $\bU\sim C=C_{\btheta_0}$,
$\disp B_{A,ji} \le M \kappa_{\gamma,\cN} E\left\{ \I\left(\bU \in [\gamma,1-\gamma]^\complement \right)\; q_{A^\complement}(\bU) \right\}$. 
By hypothesis,  $q_{A^\complement}$ is integrable with respect to $C$.
Also, by hypothesis \eqref{eq:H4}, for $\bU\sim C$,
$\disp 
B_{A,j3} \le M \kappa_{\gamma,\cN} E\left\{ \I\left(\bU \in [\gamma,1-\gamma]^\complement \right) r_0^\epsilon (U_j)
 q_{j,A^\complement}(\bU)\right\}$, 
and by hypothesis, one can find $\epsilon\in (0,0.5)$ so that for any $A\subset \setd$ and $j\in A^\complement$,
$r_0^\epsilon (U_j) q_{j,A^\complement}(\bU) $ is integrable with respect to $C$. As a result, for $i\in\{1,2,3\}$,  $B_{A,j,i}< 2^{-d-2}\delta$ if $\gamma$ is small enough. Therefore, $P\left(\left|D_{1,n}\right|> \delta\right)<\delta$ by choosing $\epsilon$, $M$, and $\gamma$ appropriately.
Next,
set
$$
D_{2n} = \sup_{\btheta\in\cN} \sn \left|\int_{S_\gamma^\complement} (H_{n,\btheta}-H_\btheta) d\cK_n\right| \le \sum_{A\subset \setd} \sum_{j=1}^5 D_{2,A,nj},
$$
with
\begin{eqnarray*}
D_{2,A,n1} &=& \dfrac{1}{n}\sum_{j\in A} \sum_{i=1}^nJ_A(\bX_i)\I(\bX_i\not \in S_\gamma)|\dF_{nj}(X_{ij})|\sup_{\btheta\in\cN} \left|\partial_{b_j} \bvarphi_{A,\btheta}\left(\tilde\bv_{n,i}, \tilde \bu_{n,i}\right)\right|\I\{F_{j}(X_{ij}-)=0\},\\
D_{2,A,n2} &=& \dfrac{1}{n}\sum_{j\in A} \sum_{i=1}^nJ_A(\bX_i) \I(\bX_i\not \in S_\gamma)|\dF_{nj}(X_{ij})|\sup_{\btheta\in\cN} \left|\partial_{b_j} \bvarphi_{A,\btheta}\left(\tilde\bv_{n,i},\tilde \bu_{n,i}\right)\right| \I\{F_{j}(X_{ij})<1\}\\
D_{2, A,n3} &=& \dfrac{1}{n}\sum_{j\in A} \sum_{i=1}^nJ_A(\bX_i) \I(\bX_i\not \in S_\gamma)|\dF_{nj}(X_{ij}-)|\sup_{\btheta\in\cN} \left|\partial_{a_j} \bvarphi_{A,\btheta}\left(\tilde\bv_{n,i}, \tilde \bu_{n,i}\right)\right| \I\{F_{j}(X_{ij})=1\}\\
D_{2,A,n4} &=& \dfrac{1}{n}\sum_{j\in A} \sum_{i=1}^nJ_A(\bX_i)\I(\bX_i\not \in S_\gamma) |\dF_{nj}(X_{ij}-)|\sup_{\btheta\in\cN} \left|\partial_{a_j}\bvarphi_{A,\btheta}\left(\tilde \bv_{n,i}, \tilde \bu_{n,i} \right)\right|\I\{F_{j}(X_{ij}-)>0\},\\
D_{2,A,n5} &=& \dfrac{1}{n}\sum_{j\in A^\complement} \sum_{i=1}^nJ_A(\bX_i) \I(\bX_i\not \in S_\gamma)|\dF_{nj}(X_{ij})|\sup_{\btheta\in\cN} \left| \partial_{u_j}\bvarphi_{A,\btheta}\left(\tilde \bv_{n,i}, \tilde \bu_{n,i} \right)\right|,
\end{eqnarray*}
where, for any $i\in\setn$, there exists $\lambda_i \in (0,1)$ such that $\tilde \bu_{n,i} = \lambda_i \bF(\bX_i)+(1-\lambda_i)\bF_n(\bX_i)$ and
 $\tilde \bv_{n,i} = \lambda_i \bF(\bX_i-)+(1-\lambda_i)\bF_n(\bX_i-)$.
Now, for $\delta>0$, one can $\omega \in (0,1)$ such that
$
P(\cD_n)>1-\delta/2
$, where
$\disp \cD_n = \cap_{j=1}^d \cap_{i=1}\left\{ \omega F_j(X_{ij}) \le F_{nj}(X_{ij})\le 1-\omega\bar F_j(X_{ij})\right\}$,
and $\bar F_j(x_j)=1-F_j(x_j)$, $j\in\setd$. As a result, on $\cD_n$,
$ \omega \bF(\bX_{i}) \le \tilde \bu_{ni}\le 1-\omega \bar \bF(X_{i})$. It then follows that
$ \omega \bF(\bX_{i}-) \le \tilde \bv_{ni}\le 1-\omega \bar \bF(X_{i}-)$.
Hence, since $q_{j,A^\complement}\in \cQ_{A^\complement}$, for any $j\in\setd$ with $X_{ij}\not\in\cA_j$, one has
$q_{j,A^\complement}(\tilde u_{n,i}) \le c(\omega) q_{j,A^\complement}(U_j)$. It then follows that on $\cD_n \cap \cB_{n,M,\epsilon}$,
$\disp 
\sum_{j=1}^4 D_{2,A,nj} \le c(\omega) M \kappa_{\gamma,\cN} \frac{1}{n}\sum_{i=1}^n \I\left(\bU_i \in [\gamma,1-\gamma]^\complement \right) q_{A^\complement}(\bU_{i})$, 
which can be made arbitrarily small by the strong law of large numbers since  $q_{A^\complement}(\bU)$ is integrable with respect to $C$. Finally,
$\disp 
D_{2,A,n5} \le M \kappa_{\gamma,\cN} \frac{1}{n}\sum_{i=1}^n  \I\left(\bU_i \in [\gamma,1-\gamma]^\complement \right) \sum_{j\in A^\complement} r_0^\epsilon (U_{ij})  q_{j,A^\complement}(\bU_{i})$,
and by hypothesis, one can find $\epsilon \in (0,0.5)$ so that for any $A\subset \setd$ and $j\in A^\complement$,
$r_0^\epsilon (U_j) q_{j,A^\complement}(\bU) $ is integrable with respect to $C$. One may then conclude that for $\delta>0$, one can find $\omega,\gamma, M,\cN$ so that $P(|D_{2,n}|> \delta)<\delta$. The rest of the proof follows easily.
\end{proof}

The following assumption is needed for proving convergence of estimators.

\begin{ass}\label{hyp:est} There exists a neighborhood $\cN$ of $\btheta_0$ such that
$\bmu(\btheta_0)=0$  and the Jacobian $\dot \bmu(\btheta)$ exists, is continuous, and is positive definite or negative definite at $\btheta_0$.
\end{ass}
\begin{remark}\label{rem:est} From Proposition \ref{prop:rolle},  
 there is a neighborhood $\cN$ of $\btheta_0$ on which $\bmu$ is injective, and
 $\dot \bmu(\btheta)$ is either positive definite or negative definite for all $\btheta\in\cB$. In particular, $\bmu$ has  a unique zero in $\cN$.
\end{remark}

We can now prove the following general convergence result for rank-based Z-estimators.
\begin{thm}\label{thm:gen_est}
Under Assumptions \ref{hyp:varphiApp} and \ref{hyp:est}, if $\btheta_n$ satisfies $\bT_n(\btheta_n) =0$, then $\bTheta_n = \sn(\btheta_n-\btheta_0)$ converges in law to $\bTheta =  -\{\dot \bmu(\btheta_0)\}^{-1}\dT(\btheta_0)$.
\end{thm}
\begin{proof}
Suppose that $\btheta_n$ is such that $\bT_n(\btheta_n)=0$ and suppose that  $\btheta_{n_k}$ is a sub-sequence converging to $\btheta^\star$. Such a sub-sequence exists by choosing a bounded neighborhood $\cN$ of $\btheta_0$. Then, on one hand, $\bmu\left(\btheta_{n_k}\right) \to
\bmu\left(\btheta^\star\right)$. On the other hand, from  Theorem \ref{thm:rum_ext},
$ \bmu\left(\btheta_{n_k}\right)  =
-\frac{\dT_{n_k}(\btheta_{n_k})}{\sqrt{n_k}} \stackrel{Pr}{\to} 0$.
As a result, $\bmu\left(\btheta^\star\right)=0$ and by Remark \ref{rem:est}, $\btheta^\star=\btheta_0$. Since every subsequence converges to the same limit, it follows that $\btheta_n$  converges in probability to $\btheta_0$. Using  Theorem \ref{thm:rum_ext} again, on one hand, one may conclude that
$
-\sn\bmu(\btheta_n) = \dT_n(\btheta_n)\rightsquigarrow \dT(\btheta_0)$. On the other hand, for some $\lambda_n\in [0,1]$, one gets that
$$
\sn\bmu(\btheta_n) =  \sn\left\{\bmu(\btheta_n)-\bmu(\btheta_0)\right\} = \dot \bmu(\tilde\btheta_n)\bTheta_n,
$$
where $\tilde \btheta_n = \btheta_0+ \lambda_n (\btheta_n-\btheta_0)\stackrel{Pr}{\to}\btheta_0$. By Assumption \ref{hyp:est}, $\dot \bmu(\tilde\btheta_n)\stackrel{Pr}{\to}\dot \bmu(\btheta_0)$ is positive definite. As a result,
$\bTheta_n\rightsquigarrow -\{\dot \bmu(\btheta_0)\}^{-1}\dT(\btheta_0)$.
\end{proof}

\section{Application to copula families}
In the case of estimation of copula parameters, i.e., when $\bvarphi_{A,\btheta}$ is defined by formula \eqref{eq:varphi-cop},
set $ \cW_1=\dT_1(\btheta_0)$ and $\cW_2=\dT_2=(\btheta_0)$.
For any $j\in\setd$, define $  \disp \cI_j = \cup_{y\in \cA_j}\cI_j(y)$, and $\cI_j(y) = (F_j({y} -), F_j(y)]$. Next, for $x_j\in \cA_j$, $j\in A$, set
$\disp G_{A,\btheta}(\bx,\bu)= \int_{(0,1)^d} \left[\I\{v_j\in \cI_j(x_j)\right]c_\btheta\left((\bv,\bu)^A \right)d\bv$,
where $(\bv,\bu)^A_j = v_j$ if $j\in A$ and $(\bv,\bu)^A_j = u_j$ if $j\in A^\complement$. Then
\begin{equation}\label{eq:GA}
G_{A,\btheta}(\bx,\bu)=  \sum_{B\subset A} (-1)^{|B|} \partial_{A^\complement} C_\btheta\left\{(\bx,\bu)^{(B,A)}\right\},
\end{equation}
where
$
\left(( \bx,\bu)^{(B,A)}\right)_j =
\left\{
\begin{array}{ll}
 F_j({x_j} -) ,& j\in  B; \\
  F_j(x_j),   & j  \in  A\setminus B;\\
 u_j , & j\in A^\complement.
 \end{array}
 \right. $. \\
Note that if $\tilde G_{A,\btheta_0}(\bx,\bu) =  G_{A,\btheta_0}(\bx,\bu)\left[\prod_{k\in A^\complement}\I\{u_k\not\in \cI_k\}\right]$, then
using the previous notations, one gets
\begin{multline*}
\int \dH_{A,\btheta} d\cK = \sum_{j\in A} \sum_{ \disp \bx \in  \bigtimes_{k\in A}\cA_k} \dB_j\circ F_j(x_j)\int \partial_{b_j}\bvarphi_{A,\btheta}\{\bF_A(\bx-),\bF_A(\bx),\bu\}\tilde G_{A,\btheta_0}(\bx,\bu)d\bu  \\
 + \sum_{j\in A}\sum_{ \disp \bx \in  \bigtimes_{k\in A}\cA_k} \dB_j\circ F_j(x_j-)\int \partial_{a_j}\bvarphi_{A,\btheta}\{\bF_A(\bx-),\bF_A(\bx),\bu\} \tilde G_{A,\btheta_0}(\bx,\bu) d\bu \\
 + \sum_{j\in A^\complement}\sum_{ \disp \bx \in  \bigtimes_{k\in A}\cA_k} \int \dB_j(u_j)  \partial_{u_j}\bvarphi_{A,\btheta}\{\bF_A(\bx-),\bF_A(\bx),\bu\} \tilde  G_{A,\btheta_0}(\bx,\bu) d\bu.
\end{multline*}
Before stating the next result, we define additional functions that will appear in the expression of $\dT_2$.

\subsection{Auxiliary functions}\label{app:eta}

For  $\bu = (u_1,\ldots,u_d)$, $\bx=(x_1,\ldots,x_d)$, and $j\in A^\complement$, define
$J_A(\bu,\bx) = \left[ \prod_{k\in A}\I\{ u_k \in \cI_k(x_k)\}\right] \left[ \prod_{k\in A^\complement}\I\{ u_k\not\in \cI_k\}\right] $, and set
$\bbeta_{A,j}(\bu) =\bbeta_{A,j,1}(\bu) -\bbeta_{A,j,2}(\bu) $, where
\begin{eqnarray*}\label{eq:etaiA1}
\bbeta_{A,j,1}(\bu) &=&
\sum_{ \disp \bx \in  \bigtimes_{k\in A}\cA_k}  \dfrac{\partial_{u_j}\dot G_{A,\btheta_0}(\bx,\bu) }{G_{A,\btheta_0}(\bx,\bu)}  J_A(\bu,\bx)
  ,\\
\label{eq:etaiA2}
\bbeta_{A,j,2}(\bu) &=&
\sum_{ \disp \bx \in  \bigtimes_{k\in A}\cA_k}  \dfrac{\dot G_{A,\btheta_0}(\bx,\bu)\partial_{u_j} G_{A,\btheta_0}(\bx,\bu)}{G_{A,\btheta_0}^2(\bx,\bu)}  J_A(\bu,\bx) .
\end{eqnarray*}
Next, for $j\in A$ and $x_j\in \cA_j$, define $\bbeta_{A,j,+}(x_j,\bu)= \bbeta_{A,j,1+}(x_j,\bu)-\bbeta_{A,j,2+}(x_j,\bu) $ and $\bbeta_{A,j,-}(x_j,\bu) = \bbeta_{A,j,1-}(x_j,\bu)-\bbeta_{A,j,2-}(x_j,\bu)$, where
\begin{eqnarray*}
\bbeta_{A,j,1+}(x_j,\bu) & =&
\sum_{ \disp \bx \in  \bigtimes_{k\in A\setminus\{j\}}\cA_k} \dfrac{\disp  \sum_{B\subset A\setminus\{j\}} (-1)^{|B|} \partial_{u_j}\partial_{A^\complement}\dot C_{\btheta_0}\left((\bx,\bu)^{(B,A)}\right)}{G_{A,\btheta_0}(\bx,\bu)}   J_A(\bu,\bx)  ,\\
\end{eqnarray*}
\begin{eqnarray*}
\bbeta_{A,j,2+}(x_j,\bu) & =&
\sum_{ \disp \bx \in  \bigtimes_{k\in A\setminus\{j\}}\cA_k}  \dfrac{L_A(\bx,\bu)}{G_{A,\btheta_0}(\bx,\bu)}  \sum_{B\subset A\setminus\{j\}} (-1)^{|B|} \partial_{u_j}\partial_{A^\complement} C_{\btheta_0}\left((\bx,\bu)^{(B,A)}\right)  \\
&& \qquad \qquad \times J_A(\bu,\bx)  ,
\end{eqnarray*}
\begin{eqnarray*}
\bbeta_{A,j,1-}(x_j,\bu) & =& \sum_{ \disp \bx \in  \bigtimes_{k\in A\setminus\{j\}}\cA_k}  \dfrac{\disp \sum_{B\subset A\setminus\{j\}} (-1)^{|B|} \partial_{u_j}\partial_{A^\complement}\dot C_{\btheta_0}\left((\bx,\bu)^{(B\cup\{j\},A)}\right)} {G_{A,\btheta_0}(\bx,\bu)} J_A(\bu,\bx)  ,
\end{eqnarray*}
\begin{eqnarray*}
\bbeta_{A,j,2-}(x_j,\bu) & =& \sum_{ \disp \bx \in  \bigtimes_{k\in A\setminus\{j\}}\cA_k}  \dfrac{L_A(\bx,\bu)}{G_{A,\btheta_0}(\bx,\bu)} \sum_{B\subset A\setminus\{j\}} (-1)^{|B|} \partial_{u_j}\partial_{A^\complement} C_{\btheta_0}\left((\bx,\bu)^{(B\cup\{j\},A)}\right) \\
&& \qquad\qquad \times J_A(\bu,\bx)  ,
\end{eqnarray*}
with $\bL_A =  \dfrac{\dot G_{A,\btheta_0}}{G_{A,\btheta_0}}$.

\subsection{Asymptotic behavior of the estimator}\label{app:zeta}

Recall that $\bzeta_{A,i} = H_{A,\btheta_0}(\bX_i)$ and $\bzeta_{i} = H_{\btheta_0}(\bX_i)$, $i\in\setn$.
\begin{thm}\label{thm:zeta}
Set $\disp \cW_{n,0} = \sni \sum_{i=1}^n \frac{\dot c_{\btheta_0}(\bU_i)}{c_{\btheta_0}(\bU_i)}$. Under Assumptions \ref{hyp:varphiApp} and  \ref{hyp:est}, the iid  random vectors $\bzeta_i$ have mean $0$ and  covariance $\cJ_1 = E\left(\bzeta_i \bzeta_i^\top\right)$.  In particular, $\bmu(\btheta_0)=0$. Furthermore,  $(\cW_{n,0}, \cW_{n,1},\cW_{n,2})$ converges in law to a centered Gaussian vector $(\cW_0, \cW_1,\cW_2)$ with $\cW_1 \sim N(0,\cJ_1)$, 
where $\cW_2$ given by \begin{eqnarray*}
\cW_{2} &=&  -\sum_{j=1}^d \sum_{A\not\ni j} \int_{(0,1)^d} c_{\btheta_0}(\bu) \bbeta_{A,j,2}(\bu)\dB_{j}(u_{j})d\bu\\
&& \quad -  \sum_{j=1}^d \sum_{x_j\in \cA_j} \dB_{j}\{F_j(x_j)\} \sum_{A\ni j}  \int_{(0,1)^d} c_{\btheta_0}(\bu)  \bbeta_{A,j,2+}(x_j,\bu)d\bu\\
&& \qquad
 + \sum_{j=1}^d \sum_{x_j\in \cA_j} \dB_{j}\{F_j(x_j-)\} \sum_{A\ni j} \int_{(0,1)^d} c_{\btheta_0}(\bu)  \bbeta_{A,j,2-}(x_j, \bu)d\bu.
\end{eqnarray*}
Finally, $E\left(\cW_0 \cW_1^\top \right)  = \cJ_1 =-\dot \bmu(\btheta_0)$, and $E\left(\cW_0 \cW_2^\top \right)  = 0 $.
In particular, Assumption \ref{hyp:est} holds 
if $\cJ_1$ is invertible.
\end{thm}
\begin{proof}
For  any $A\subset \setd$, set $\cI_A(\bu) = \left[ \prod_{j\in A}\I\{u_j\in \cI_j\}\right] \left[ \prod_{k\in A^\complement}\I\{u_k \not\in \cI_k\}\right]$.
It follows from formula \eqref{eq:GA} that
\begin{eqnarray*}
 E(\bzeta_{A,1}) &=&    \sum_{ \disp \bx \in  \bigtimes_{j\in A}\cA_j}\int c_{\btheta_0}(\bu)\left[\prod_{k\in A}\I\{u_k\in\cI_k(x_k)\}\right]
 \left[\prod_{k\in A^\complement}\I(u_{k}\notin\cI_k)\right] \dfrac{\dot G_{A,\btheta_0}(\bx,\bu)}{G_{A,\btheta_0}(\bx,\bu)}d\bu\\
&=& \sum_{ \disp \bx \in  \bigtimes_{j\in A}\cA_j}\int \left[\prod_{k\in A^\complement}\I(u_{k}\notin\cI_k)\right] \dot G_{A,\btheta_0}(\bx,\bu) d\bu\\
&=& \int \dot c_{\btheta_0}(\bu) \cI_A(\bu) d\bu =\left.\nabla_\btheta E_\btheta \{\cI_A(\bU)\}\right|_{\btheta=\btheta_0}.
\end{eqnarray*}
Thus,
$\disp 
\bmu(\btheta_0) = E(\bzeta_1) = \sum_{A\subset \setd}E(\bzeta_{A,1}) = \left.  \nabla_\btheta E_\btheta\left\{\sum_{A\subset \setd}\cI_A(\bu) \right\}\right|_{\btheta=\btheta_0} = \left. \nabla_\btheta\left[1 \right]\right|_{\btheta=\btheta_0}= 0$.
As a result, the independent random variables $\bzeta_i$  are centered, and we may conclude that $\dT_{1,n}(\btheta_0) = n^{-1/2} \sum_{i=1}^n \bzeta_i \rightsquigarrow \dT_1(\btheta_0) \sim N(0,\cJ_1)$, where
$\disp 
\cJ_1 = E\left(\bzeta_1 \bzeta_1^\top\right)=\sum_{A\subset \setd} E\left(\bzeta_{A,1} \bzeta_{A,1}^\top\right)$. 
Also, $\cW_{n,2} = \dT_{n,2}(\btheta_0)  \rightsquigarrow  \cW_{2}=\dT_{2}(\btheta_0)$ since  $\dB_{n1},\ldots,\dB_{nd}$ converge jointly to the Brownian bridges $\dB_{1},\ldots,\dB_{d}$.
Now
\begin{multline*}
\int \dH_{A,\btheta_0} d\cK =  \sum_{j\in A}\sum_{ \disp \bx \in  \bigtimes_{k\in A}\cA_k} \dB_{j}\circ F_j(x_j)\int c_{\btheta_0}(\bu) \{\bbeta_{A,j,1+}(x_j,\bu)-\bbeta_{A,j,2+}(x_j,\bu)\}d\bu  \\
- \sum_{j\in A}\sum_{ \disp \bx \in  \bigtimes_{k\in A}\cA_k} \dB_{j}\circ F_j(x_j-)
 \int c_{\btheta_0}(\bu) \{\bbeta_{A,j,1-}(x_j,\bu)-\bbeta_{A,j,2-}(x_j,\bu)\}d\bu\\
 + \sum_{j\in A^\complement}\sum_{ \disp \bx \in  \bigtimes_{k\in A}\cA_k} \int \dB_{j}(u_j)  c_{\btheta_0}(\bu) \{\bbeta_{A,j,1}(\bu)-\bbeta_{A,j,2}(\bu)\}d\bu.
\end{multline*}
Using the same trick as in the proof that $E(\bzeta_i)=0$, one gets, for any $j\in\setd$, and $x_j\in \cA_j$
$$
  \sum_{A\not\ni j} \int_{(0,1)^d} c_{\btheta_0}(\bu) \bbeta_{A,j,1}(\bu)\dB_{j}(u_{j})d\bu= 0=   \sum_{A\ni j}  \int_{(0,1)^d} c_{\btheta_0}(\bu)  \bbeta_{A,j,1\pm}(x_j,\bu)d\bu,
$$
yielding the representation for $\cW_2$.
Next, it is then easy to show that $E\{\cW_0 \dB_j(t)\} = 0$ for any $j\in\setd$, so $E\left(\cW_0\cW_2^\top \right)=0$.
Next, for any $A\subset \setd$,
\begin{eqnarray*}
E\left\{\bzeta_{A,i}\dfrac{\dot c_{\btheta_0}(\bU_i)^\top }{c_{\btheta_0}(\bU_i)}\right\} &=&  \sum_{ \disp \bx \in  \bigtimes_{j\in A}\cA_j}\int \prod_{k\notin A}\I(u_{ik}\notin\cI_k)\dfrac{\dot G_{A,\btheta_0}(\bx,\bu) \dot G_{A,\btheta_0}(\bx,\bu)^\top}{G_{A,\btheta_0}(\bx,\bu)}d\bu=E\left(\bzeta_{A,i} \bzeta_{A,i}^\top \right).
\end{eqnarray*}
Therefore, $\left(\cW_1 \cW_0^\top \right) = \cJ_1$.
Next, for any $A\subset\setd$,
$$
\bmu_A(\btheta) = \sum_{ \disp \bx \in  \bigtimes_{j\in A}\cA_j}\int \prod_{k\notin A}\I(u_{ik}\notin\cI_k)\dfrac{\dot G_{A,\btheta}(\bx,\bu) }{G_{A,\btheta}(\bx,\bu)}G_{A,\btheta_0}(\bx,\bu)d\bu .
$$
Hence, for any $A\subset\setd$,
\begin{eqnarray*}
\dot \bmu_A(\btheta) &=& \sum_{ \disp \bx \in  \bigtimes_{j\in A}\cA_j}\int \prod_{k\notin A}\I(u_{ik}\notin\cI_k)\dfrac{\ddot G_{A,\btheta}(\bx,\bu) }{G_{A,\btheta}(\bx,\bu)}G_{A,\btheta_0}(\bx,\bu)d\bu\\
&& \qquad - \sum_{ \disp \bx \in  \bigtimes_{j\in A}\cA_j}\int \prod_{k\notin A}\I(u_{ik}\notin\cI_k)\dfrac{\dot G_{A,\btheta}(\bx,\bu) \dot G_{A,\btheta}(\bx,\bu)^\top }{G_{A,\btheta}^2(\bx,\bu)}G_{A,\btheta_0}(\bx,\bu)d\bu,
\end{eqnarray*}
so
\begin{eqnarray*}
\dot \bmu_A(\btheta_0) &=& \sum_{ \disp \bx \in  \bigtimes_{j\in A}\cA_j}\int \prod_{k\notin A}\I(u_{ik}\notin\cI_k) \ddot G_{A,\btheta}(\bx,\bu) d\bu\\
&& \qquad - \sum_{ \disp \bx \in  \bigtimes_{j\in A}\cA_j}\int \prod_{k\notin A}\I(u_{ik}\notin\cI_k)\dfrac{\dot G_{A,\btheta_0}(\bx,\bu) \dot G_{A,\btheta_0}(\bx,\bu)^\top }{G_{A,\btheta_0}(\bx,\bu)}d\bu\\
&=& \left. \nabla^2_\btheta \left\{ \int c_\btheta(\bu) \cI_A(\bu)d\bu  \right\}\right|_{\btheta=\btheta_0} - E\left(
\bzeta_{A,1}\bzeta_{A,1}^\top\right).
\end{eqnarray*}
As a result,
$\disp \dot \bmu(\btheta_0)=\sum_{A\subset \setd} \dot \bmu_A(\btheta_0) = 0 - \cJ_1 = -\cJ_1$.
\end{proof}

\section{Supplementary material}
Here, some results necessary to identifiability are proven, together with multilinear representations and results on densities with respect to product of mixtures of Lebesgue's measure and counting measure, as well as a result on the verification of Assumption 5.

\subsection{Proof of  Proposition 2.1}
Suppose that the maximal rank of $\bT'$ is $r$ and let $\btheta_0\in O$ so that the rank of $\bT'(\btheta_0)$ is $r$. Set $\bA = \bT'(\btheta_0)$. Then there  exists an $r\times r$ submatrix $M_r(\btheta_0)$ with non-zero determinant. By continuity, there is a neighborhood $\cN$ of $\btheta_0$ such that  $M_r(\btheta)$ has non-zero determinant. Since the maximal rank is $r$, it follows that the rank of $\bT'(\btheta)$ is $r$ for all $\btheta\in\cN$. One can now apply the famous Rank Theorem \citep[Theorem 9.32]{Rudin:1976} to deduce that there is a diffeomorphism $\bH$ from an open set $V\subset \dR^p$ onto  a neighborhood $U\subset \cN$, $\btheta_0\in U$, such that $\bT\{\bH(\bx)\} = \bA\bx + \bvarphi(\bA\bx)$, where $\bvarphi$ is a continuously differentiable mapping into the null space of a projection withe same range as $A$. Thus, choose $\bx_1$ in the null space of $A$ and $\bx_0\in V$ so that $\bH(\bx_0)=\btheta_0$. Then, for some $\delta>0$,  $\bH(\bx_0+t\bx_1)\in U$ for any $|t|<\delta$. As a result, for all $|t|<\delta$,
 $  \bT\{\bH(\bx_0+t\bx_1)\} = A(\bx_0+t\bx_1)+\bvarphi(A\bx_0+tA \bx_1) =
 A\bx_0+ \bvarphi(A\bx_0)$,
 showing that $T$ is not injective.
Next, suppose that $\bJ$ has rank $p$ and $\bT$ is not injective, i.e, one can find $\btheta_0\neq \btheta_1 \in O$, with $T(\btheta_0)=T(\btheta_1)$. Set $g(t) = T_n\{\btheta_0+t ( \btheta_1-\btheta_0)\}$. This is well defined on $[0,1]$ since $O$ is convex. Since $g(0)=g(1)$ and $g$ is continuous, $g([0,1])$ is a closed curve and there exist $t_0,t_1 \in [0,1]$ with $g_j(t_0)\le g_j(t)\le g_j(t_1)$, for all $j\in\setd$ and for all $t\in[0,1]$. It then follows that either $t_0\neq 0$ or $t_1\neq 0$. Otherwise $g$ is constant and $g'(t)=0$ for all $t\in [0,1]$. So suppose that $t_0 >0$; the case $t_1>0$ is similar. Then $0<t_0<1$, which implied that for any $j\in \setd$ and for $\delta>0$ small enough, $\dfrac{g_j(t_0-\delta)-g_j(t_0)}{\delta} \ge 0 \ge \dfrac{g_j(t_0+\delta)-g_j(t_0)}{\delta} $. As a result, $g_j'(t_0)=0$ for all $j\in\setd$. Therefore, from this extension of  Rolle's Theorem,
$0 = g'(t) = \bJ\{\btheta_0+t ( \btheta_1- \btheta_0)\}(\btheta_1-\btheta_0)$.
As a result, the Jacobian has not rank $p$ at $\btheta=\btheta_0+t ( \btheta_1- \btheta_0)\in O$, which is a contradiction. Hence $\bT$ is injective.  Suppose now that the rank of $\bJ(\btheta_0)$ is $p$. Set $g(\btheta) = {\rm det}\left\{\bJ(\btheta)^\top \bJ(\btheta)\right\} \ge 0$. Then the rank of $\bJ(\btheta)$ is $p$ iff $g(\btheta)>0$. Since $g(\btheta_0)>0$ and $g$ is continuous on $O$, one can find a neighborhood $\cN$ of $\btheta_0$ for which $g(\btheta)>0$ for every $\btheta\in \cN$. Finally, if the maximal rank is $r<p$, then if follows from the proof of the first statement that the rank of $\bT'$ is also $r$ in a neighborhood of $\btheta_0$, and from the Rank Theorem  \citep[Theorem 9.32]{Rudin:1976} , $\bT$ is not injective in a neighborhood of $\btheta_0$.
\qed

\subsection{Density with respect to  $\mu$}
First, for a given cdf $\cG$ having density $g$ with respect to $\mu_{\cA}+\cL$, where $\cA$  is a countable set, one has $g(x) = \Delta \cG(x)$, for $x\in \cA$, and
\begin{equation}\label{eq:dens0}
\tilde \cG(x) = P(X\le x, X\not\in\cA) = \cG(x)-\sum_{y\in \cA, \; y\le x}\Delta \cG(y) =\int_{-\infty}^x  g(y)dy.
\end{equation}
Then, by the fundamental theorem of calculus, $\tilde \cG'(x)=g(x)$, $\cL$-a.e.   If in addition $\cA$ is closed, then for any $x\not \in \cA$, and $\epsilon>0$ small enough, $\cG(x+\epsilon)-\cG(x) = \int_x^{x+\epsilon}g(y)dy$, so $\cG'(x) = g(x)$ a.e.
Next, to find the expression of the density $h$ of distribution function $H$ with respect to $\mu$, suppose that $x_1,\ldots, x_k$ are atoms of $F_1,\ldots, F_k$, while $x_{k+1},\ldots, x_d$ are not atoms of $F_{k+1},\ldots, F_d$.
On one hand, if $\bX\sim H$, for $\delta>0$ small enough, one gets
\begin{eqnarray*}
G(\bx,\delta) &=& P\left[ \cap_{j=1}^k \{X_j=x_j\} \cap \cap_{j=k+1}^d \{x_j-\delta  < X_j \le x_j\}\cap\{X_j\not\in \cA_j\} \right] \\
& = & \int h(x_1,\ldots,x_k,y_{k+1},\ldots,y_d) \left\{\prod_{j=k+1}^d \I(x_j-\delta < y_j \le x_j)  \right\} dy_{k+1}\cdots dy_d.
\end{eqnarray*}
On the other hand, if $\bU\sim C$, with $C\in \cS_H$ having a continuous density over $(0,1)^d$,
using the fact that for any $j>k$, the Lebesgue's measure of
$(F_j(x_j-\delta),F_j(x_j)]\setminus \cI_j$ is $\int_{x_j-\delta}^{x_j} f_j(y_j)dy_j$ by \eqref{eq:dens0}, one gets
\begin{eqnarray*}
G(\bx,\delta) & = &  P\left[\cap_{j=1}^k \{F_j(x_j-)< U_j\le F_j(x_j)\} \cap  \cap_{j=k+1}^d \{F_j(x_j-\delta_j) < U_j \le F_j(x_j)\}  \cap \{U_j\not\in\cI_j\}\right] \\
&=& \int_{F_1(x_1-)}^{F_1(x_1)} \cdots \int_{F_k(x_k-)}^{F_k(x_k)} \int_{ (F_{k+1}(x_{k+1}-\delta),F_{k+1}(x_{k+1})]\setminus \cI_{k+1}}\cdots \int_{ (F_d(x_d-\delta_d),F_d(x_d)]\setminus \cI_d}c(\bu)d\bu\\
&=& \int_{F_1(x_1-)}^{F_1(x_1)} \cdots \int_{F_k(x_k-)}^{F_k(x_k)}
c\left( u_1,\ldots,u_k,F_{k+1}(x_{k+1}),\ldots, F_d(x_d)\right) d u_1 \cdots d u_k  \\
&& \qquad \qquad \qquad \times \left\{ \prod_{j=k+1}^d \int_{x_j-\delta}^{x_j} f_j(y_j) d_j \right\} +o\left(\delta^{d-k}\right).
\end{eqnarray*}
Therefore,
\begin{eqnarray*}
h(\bx) &=& \prod_{j=k+1}^d  f_j(x_j) \int_{F_1(x_1-)}^{F_1(x_1)} \cdots \int_{F_k(x_k-)}^{F_k(x_k)}
c\left( u_1,\ldots,u_k,F_{k+1}(x_{k+1}),\ldots, F_d(x_d)\right) d u_1 \cdots d u_k   \nonumber\label{eq:dens1}\\
&=& \prod_{j=k+1}^d f_j(x_j) \sum_{B\subset \{1,\ldots,k\}}(-1)^{|B|} \partial_{u_{k+1}}\cdots \partial_{u_{d}} C\left(\tilde\bF^{(B)}(\bx)\right), \label{eq:dens2}
\end{eqnarray*}
where $\tilde F^{(B)}$ is defined by
\begin{equation}\label{eq:FB}
\left(\tilde \bF^{(B)}(\bx)\right)_j =
\left\{
\begin{array}{ll}
F_{j}({x_{j}}-)  & \text{, if } j \in B;\\
 F_{j}(x_{j})   & \text{, if } j \in B^\complement = \setd\setminus B.
 \end{array}\right.
 \end{equation}
As  a result, if $A = \{j: x_j \in \cA_j\}$, then the density $h(x)$ with respect to $\mu$ is
\begin{equation}\label{eq:density}
 h(\bx) =  \left\{  \prod_{j\in A^\complement} f_j(x_j)\right\} \sum_{B\subset A }(-1)^{|B|} \partial_{A^\complement} C \left\{\tilde\bF^{(B)}(\bx)\right\}.
\end{equation}
In particular, if $A$ is the set of indices $j\in\setd$ for which $\Delta F_j(x_j)>0$, and $C$ is a copula associated with the joint law of $(Y,\bX)$, where $Y$ has margin $G$, then \eqref{eq:density} yields
\begin{equation}\label{eq:RegressionCopGen}
P(Y\le y|\bX=\bx) = \dfrac{ \sum_{B\subset A} (-1)^{| B|} \partial_{A^\complement} C\left\{G(y),\tilde\bF^{(B)}(\bx)\right\}}{ \sum_{B\subset A} (-1)^{| B|} \partial_{A^\complement} C\left\{1,\tilde\bF^{(B)}(\bx)\right\}}.
\end{equation}


\subsection{Verification of Assumption 5  }
Let $G$ be a discrete distribution function concentrated on $\cA = \{x: \Delta G(x)>0\}$. Recall  that
$\disp V_n(G) = n\sum_{x\in \cA}\Delta G(x)\{1-\Delta G(x)\}^{n-1}$.
\begin{prop}\label{prop:verif1}
Suppose that there exists a constant $C$ so that for any $a$ small enough,
\begin{equation}\label{eq:cond1}
\sum_{x\in \cA: \Delta G(x) <a} \Delta G(x)  \le C a.
\end{equation}
Then $\limsup_{n\to\infty}V_n(G)< \dfrac{C}{\left(1-e^{-1}\right)^2}$.
\end{prop}
\begin{proof}
First, recall that $1-x\le e^{-x}$ for any $x\ge 0$.
For simplicity, for any $x\in \cA$, set $p_x = \Delta G(x)$. Next, using \eqref{eq:cond1}, if $n$ is large enough, one has
\begin{multline*}
V_n(G) = n \sum_{k=1}^\infty \sum_{x\in \cA: k\le np_x < k+1} p_x(1-p_x)^{n-1}
 +\sum_{k=0}^\infty  \sum_{x\in \cA: n^{-(k+1)} \le p_x < n^{-k} } p_x(1-p_x)^{n-1}\\
 \le  C\sum_{k=1}^\infty (k+1) e^{-(n-1)k/n} + C\sum_{k=0}^\infty  \dfrac{1}{n^{k}} = C\left[ \dfrac{1}{\left\{1-e^{-(n-1)/n}\right\}^2} -1 + \dfrac{n}{n-1}\right].
\end{multline*}
As a result, $\limsup_{n\to\infty}V_n(G) \le \dfrac{C}{\left(1-e^{-1}\right)^2}$.
\end{proof}
Note that when  $\cA$ is finite, then \eqref{eq:cond1} holds with $C=0$.
\begin{remark}\label{rem:condverif}
Note that \eqref{eq:cond1} holds for the geometric distribution with $p_k = p(1-p)^k$, $k=0,1,\ldots$, with $p\in (0,1)$. In this case, $C = p$ since $\disp\sum_{k=i} p(1-p)^k = (1-p)^i$, $i=0,1,\ldots$. \eqref{eq:cond1} holds for Poisson distribution with parameter $\lambda$ since $p_{k+i} \le e^\lambda p_k p_i $, $k,i=0,1,\ldots$, with $p\in (0,1)$. In this case, $C = e^\lambda$. In general a sufficient condition for \eqref{eq:cond1} to hold on $\dN\cup\{0\}$ is that $p_{k+i} \le C p_k p_i $, $k,i \in \dN\cup\{0\}$.
Another example of an important distribution satisfying \eqref{eq:cond1} is the Negative  Binomial distribution with parameters $r \in \dN$ and $p \in (0,1)$, where $p_k = 
\left(\begin{array}{c} k+r-1\\r-1\end{array}\right)p^r(1-p)^k$, $k \in \dN\cup \{0\}$. In this case
$
\sum_{i \ge  k} p_i = \sum_{j=0}^{r-1}  \left(\begin{array}{c} k+r-1\\r-1\end{array}\right) p^j (1-p)^{k+r-1-j}$.
So if $a$ is small enough, $p_i$ is decreasing, so let $k_0$ be such that $p_{k_0} < a \ge p_{k_0-1}$.
It then follows that if $k_0$ is large enough,
\begin{multline*}
\sum_{i: p_i < a} p_i = \sum_{i\ge  k_0} p_i = \sum_{j=0}^{r-1}
\left(\begin{array}{c} k_0+r-1\\j\end{array}\right)
p^j (1-p)^{k_0+r-1-j} \\
\le \left(\begin{array}{c} k_0+r-1\\r-1\end{array}\right)\sum_{j=0}^{r-1}p^j (1-p)^{k_0+r-1-j} \le r  \left(\begin{array}{c} k_0+r-1\\r-1\end{array}\right) (1-p)^{k_0} = C p_{k_0} \le C a,
\end{multline*}
where $C = r p^{-r}$.
\end{remark}

\end{appendix}
\bibliographystyle{apalike} 

\end{document}